\documentclass{llncs}
\usepackage{graphicx}
\usepackage{fullpage}
\usepackage{xspace}
\usepackage[english]{babel}
\usepackage{subfigure} 
\usepackage{enumerate}

\newtheorem{openproblem}{Open Problem}

\newcommand{\cdrawing}{compatible drawing\xspace}
\newcommand{\cdrawings}{compatible drawings\xspace}

\author{Patrizio Angelini\inst{1},
Carla Binucci\inst{2},
Giordano Da Lozzo\inst{1},
Walter Didimo\inst{2},\\
Luca Grilli\inst{2},
Fabrizio Montecchiani\inst{2},
Maurizio Patrignani\inst{1},
Ioannis G. Tollis\inst{3}
}

\date{}

\institute{
Universit\`a  Roma Tre, Italy\\
\email{\{angelini,dalozzo,patrigna\}@dia.uniroma3.it} \and
Universit\`a degli Studi di Perugia, Italy\\
\email{\{binucci,didimo,grilli,montecchiani\}@diei.unipg.it} \and
Univ. of Crete and Institute of Computer Science-FORTH, Greece\\
\email{tollis@ics.forth.gr}
}

\title{Algorithms and Bounds for Drawing Non-planar Graphs\\ with Crossing-free Subgraphs\thanks{Research supported in part by the MIUR project AMANDA ``Algorithmics for MAssive and Networked DAta'', prot. 2012C4E3KT\_001. Work on these results began at the 8th Bertinoro Workshop on Graph drawing.
Discussion with other participants is gratefully acknowledged. Part of the research was conducted in the framework of ESF project 10-EuroGIGA-OP-003 GraDR "Graph Drawings and Representations" and of 'EU FP7 STREP Project "Leone: From Global Measurements to Local Management", grant no. 317647'. A preliminary extended abstract of the results contained in this paper has been presented at the 21st International Symposium on Graph Drawing, GD 2013.}}

\newenvironment{keywords}{
       \list{}{\advance\topsep by0.35cm\relax\small
       \leftmargin=1cm
       \labelwidth=0.35cm
       \listparindent=0.35cm
       \itemindent\listparindent
       \rightmargin\leftmargin}\item[\hskip\labelsep
                                     \bfseries Keywords:]}
     {\endlist}
     
\begin{document}
\maketitle

\begin{abstract}
We initiate the study of the following problem: \emph{Given a non-planar graph 
$G$ and a planar subgraph $S$ of $G$, does there exist a straight-line drawing 
$\Gamma$ of $G$ in the plane such that the edges of $S$ are not crossed in 
$\Gamma$ by any edge of $G$?}
We give positive and negative results for different kinds of connected spanning 
subgraphs $S$ of $G$.
Moreover, in order to enlarge the subset of instances that admit a solution, we 
consider the possibility of bending the edges of $G$ not in $S$; in this 
setting we discuss different trade-offs between the number of bends and the required drawing area.
\end{abstract}


\section{Introduction}\label{se:introduction}

Many papers in graph drawing address the problem of computing drawings of non-planar graphs with the goal of mitigating 
the negative effect that edge crossings have on the drawing readability.
Several of these papers describe crossing minimization methods, which are effective and computationally feasible for relatively small and sparse graphs (see~\cite{bcgjm-handbook-13} for a survey). Other papers study how non-planar graphs can be drawn such that the ``crossing complexity'' of the drawing is somewhat controlled, either in the number or in the type of crossings. They include the study of \emph{$k$-planar drawings}, in which each edge is crossed at most $k$ times (see, e.g.,~\cite{DBLP:conf/gd/BrandenburgEGGHR12,ddlm-argdf+-13,DBLP:journals/ipl/Didimo13,DBLP:conf/gd/EadesHKLSS12,DBLP:conf/cocoon/HongELP12,DBLP:journals/jgt/KorzhikM13,DBLP:journals/combinatorica/PachT97}), of \emph{$k$-quasi planar drawings}, in which no $k$ pairwise crossing edges exist (see, e.g.,~\cite{DBLP:journals/dcg/Ackerman09,DBLP:journals/jct/AckermanT07,DBLP:conf/wg/GiacomoDLM12,DBLP:journals/siamdm/FoxPS13,DBLP:journals/algorithmica/PachSS96,DBLP:journals/dcg/Valtr98}), and of \emph{large angle crossing drawings}, in which any two crossing edges form a sufficiently large angle (see~\cite{dl-cargd-12} for a survey). Most of these drawings exist only for sparse graphs.

\smallskip
In this paper we introduce a new graph drawing problem concerned with the drawing of non-planar graphs. 
Namely: \emph{Given a non-planar graph $G$ and a planar subgraph $S$ of $G$, decide whether $G$ admits a drawing $\Gamma$ such that the edges of $S$ are not crossed in $\Gamma$ by any edge of $G$, and compute $\Gamma$ if it exists.}

Besides its intrinsic theoretical interest, this problem is also of practical relevance in many application domains. Indeed, distinct groups of edges in a graph may have different semantics, and a group can be more important than another for some applications; in this case a visual interface might attempt to display more important edges without intersections. Furthermore, the user could benefit from a layout in which a connected spanning subgraph is drawn crossing free, since it would support the user to quickly recognize paths between any two vertices, while keeping the other edges of the graph visible.    

Please note that the problem of recognizing specific types of subgraphs that are not self-crossing (or that have few crossings) in a given drawing $\Gamma$, has been previously studied (see, e.g.,~\cite{jw-cdccp-93,kssw-cctg-07,kln-nstl-91,ru-sceps-13}). This problem, which turns out to be NP-hard for most different kinds of instances, is also very different from our problem. Indeed, in our setting the drawing is not the input, but the output of the problem. Also, we require that the given subgraph $S$ is not crossed by any edge of the graph, not only by its own edges.

\smallskip
In this paper we concentrate on the case in which $S$ is a connected spanning subgraph of $G$ and consider both straight-line and polyline drawings of $G$. Namely:

\smallskip\noindent{\bf {\em (i)}} In the straight-line drawing setting we prove that if $S$ is any given spanning spider or caterpillar, then a drawing of $G$ where $S$ is crossing free always exists; such a drawing can be computed in linear time and requires polynomial area (Section~\ref{sse:straight-line-trees}). We also show that this positive result cannot be extended to any spanning tree, but we describe a large family of spanning trees that always admit a solution, and we observe that any graph $G$ contains such a spanning tree; unfortunately, our drawing technique for this family of trees may require exponential area. Finally, we characterize the instances $\langle G,S \rangle$ that admit a solution when $S$ is a triconnected spanning subgraph, and we provide a polynomial-time testing and drawing algorithm, whose layouts have polynomial area (Section~\ref{sse:straight-line-triconnected}).      

\smallskip\noindent{\bf {\em (ii)}} We investigate polyline drawings where only the edges of $G$ not in $S$ are allowed to bend. In this setting, we show that all spanning trees can be realized without crossings in a drawing of $G$ of polynomial area, and we describe efficient algorithms that provide different trade-offs between the number of bends per edge and the required drawing area (Section~\ref{se:polyline}). Also, we consider the case in which $S$ is any given biconnected spanning subgraph. In this case, we provide a characterization of the positive instances, which yields drawings with non-polynomial area, if only one bend per edge is allowed, and with polynomial area if at most two bends are allowed.  

We finally remark that the study of our problem has been receiving some interest in the graph drawing community. In particular, Schaefer proved that given a graph $G$ and a planar subgraph $S$ of $G$, testing whether there exists a polyline drawing of $G$ where the edges of $S$ are never crossed can be done in polynomial time~\cite{DBLP:journals/corr/Schaefer13}. In Schaefer's setting, differently from ours, there is no restriction on the number of bends per edge and the edges of $S$ are not required to be drawn as straight-line segments.  
 
In Section~\ref{se:preliminaries} we give some preliminary definitions that will be used in the rest of the paper, while in Section~\ref{se:conclusions} we discuss conclusions and open problems deriving from our work.

\section{Preliminaries and Definitions}\label{se:preliminaries}

We assume familiarity with basic concepts of graph drawing and planarity (see, e.g., \cite{dett-gd-99}). Let $G(V,E)$ be a graph and let $\Gamma$ be a drawing of $G$ in the plane. If all vertices and edge bends of $\Gamma$ have integer coordinates, then $\Gamma$ is a \emph{grid drawing} of $G$, and the \emph{area} of $\Gamma$ is the area of the minimum bounding box of $\Gamma$.  We recall that the minimum bounding box of a drawing $\Gamma$ is the rectangle of minimum area enclosing $\Gamma$.
If $\Gamma$ is not on an integer grid, we scale it in order to guarantee the same resolution rule of a grid drawing; namely we expect that the minimum Euclidean distance between any two points on which either vertices or bends of $\Gamma$ are drawn is at least of one unit. Under this resolution rule, we define the area of the drawing as the area of the minimum bounding box of~$\Gamma$.  

Let $G(V,E)$ be a graph and let $S(V,W)$, $W \subseteq E$, be a spanning subgraph of $G$. A straight-line drawing $\Gamma$ of $G$ such that $S$ is crossing-free in $\Gamma$ (i.e., such that crossings occur only between edges of $E \setminus W$) is called a \emph{straight-line \cdrawing} of $\langle G,S \rangle$. If each edge of $E \setminus W$ has at most $k$ bends in $\Gamma$ (but still $S$ is drawn straight-line and crossing-free in $\Gamma$), $\Gamma$ is called a \emph{$k$-bend \cdrawing} of $\langle G,S \rangle$.

If $S$ is a rooted spanning tree of $G$ such that every edge of $G$ not in $S$ connects either vertices at the same level of $S$ or vertices that are on consecutive levels, then we say that $S$ is a \emph{BFS-tree} of $G$.

A \emph{star} is a tree $T(V,E)$ such that all its vertices but one have degree one, that is, $V=\{u,v_1,v_2, \dots, v_k\}$ and $E=\{(u,v_1), (u,v_2), \dots, (u,v_k)\}$; any subdivision of $T$ (including $T$), is a \emph{spider}: vertex $u$ is the \emph{center} of the spider and each path from $u$ to $v_i$ is a \emph{leg} of the spider. A \emph{caterpillar} is a tree such that removing all its leaves (and their incident edges) results in a path, which is called the \emph{spine} of the caterpillar. The one-degree vertices attached to a spine vertex $v$ are called the \emph{leaves} of $v$.

In the remainder of the paper we implicitly assume that $G$ is always a connected graph (if the graph is not connected, our results apply for any connected component).

\section{Straight-line Drawings}\label{se:straight-line}

We start studying straight-line \cdrawings of pairs $\langle G,S \rangle$: Section~\ref{sse:straight-line-trees} concentrates on the case in which $S$ is a spanning tree, while Section~\ref{sse:straight-line-triconnected} investigates the case in which $S$ is a spanning triconnected graph.

\subsection{Spanning Trees}\label{sse:straight-line-trees}

The simplest case is when $S$ is a given Hamiltonian path of $G$; in this case $\Gamma$ can be easily computed by drawing all vertices of $S$ in convex position, according to the order they occur in the path. In the following we prove that in fact a straight-line \cdrawing $\Gamma$ of $\langle G,S \rangle$ can be always constructed in the more general case in which $S$ is a spanning spider (Theorem~\ref{th:spanning-spider}), or a spanning caterpillar (Theorem~\ref{th:spanning-caterpillar}), or a BFS-tree (Theorem~\ref{th:good-trees}); our construction techniques guarantee polynomial-area drawings for spiders and caterpillars, while require exponential area for BFS-trees. On the negative side, we show that if $S$ is an arbitrary spanning tree, a straight-line \cdrawing of $\langle G,S \rangle$ may not exist (Lemmas~\ref{le:bad-trees-1} and~\ref{le:bad-trees-2}). 

\begin{theorem}\label{th:spanning-spider}
Let $G$ be a graph with $n$ vertices and $m$ edges,  and let $S$ be a spanning 
spider of $G$. There exists a grid straight-line \cdrawing $\Gamma$ of 
$\langle G,S \rangle$. Drawing $\Gamma$ can be computed in $O(n+m)$ time and has 
$O(n^3)$ area.   
\end{theorem}

\begin{figure}[tb!]
\begin{center}
\includegraphics[width=0.8\textwidth]{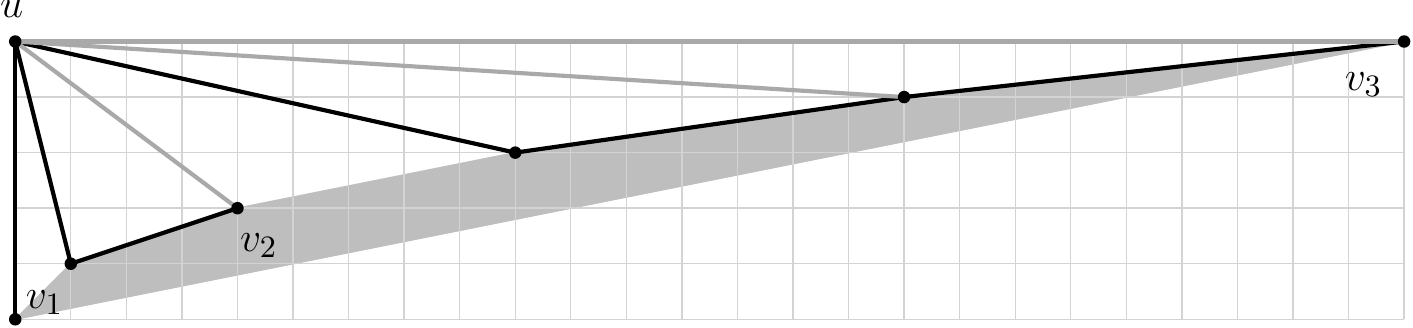}
\caption{Illustration of the drawing construction of 
Theorem~\ref{th:spanning-spider}. The thick edges belong to the spider. The 
edges of $G$ not incident to $u$ are drawn in the gray convex 
region.}\label{fi:spanning-spider}
\end{center}
\end{figure}

\begin{proof}
Let $u$ be the center of $S$ and let $\pi_1, \pi_2, \dots, \pi_k$ be the legs 
of $S$. Also, denote by $v_i$ the vertex of degree one of leg $\pi_i$ $(1 \leq i 
\leq k)$. Order the vertices of $S$ distinct from $u$ such that: $(i)$ the 
vertices of each $\pi_i$ are ordered in the same way they appear in the simple 
path of $S$ from $u$ to $v_i$; $(ii)$ the vertices of $\pi_i$ precede those of 
$\pi_{i+1}$ $(1 \leq i \leq k-1)$. If $v$ is the vertex at position $j$ $(0 \leq 
j \leq n-2)$ in the ordering defined above, draw $v$ at coordinates $(j^2,j)$. 
Finally, draw $u$ at coordinates $(0,n-2)$. Refer to 
Fig.~\ref{fi:spanning-spider} for an illustration.
With this strategy, all vertices of $S$ are in convex position, and they are 
all visible from $u$ in such a way that no edge incident to $u$ can cross other 
edges of $\Gamma$. Hence, the edges of $S$ do not cross other edges in $\Gamma$. 
The area of $\Gamma$ is $(n-2)^2 \times (n-2) = O(n^3)$ and $\Gamma$ is 
constructed in linear time.
\end{proof}

The following algorithm computes a straight-line compatible drawing of $\langle G,S 
\rangle$ when $S$ is a spanning caterpillar. 
Theorem~\ref{th:spanning-caterpillar} proves its correctness, time and area 
requirements. Although the drawing area is still polynomial, the layout is not 
a grid drawing.

The basic idea of the algorithm is as follows. It first places the spine vertices of the caterpillar in convex position, along a quarter of circumference. Then, it places the leaf vertices inside the convex polygon formed by the spine vertices, in such a way that they also suitably lie in convex position. With this strategy an edge of the caterpillar will be outside the inner polygon formed by the leaf vertices, and hence it will not cross any edge that connects two leaf vertices. Also, the inner polygon is chosen sufficiently close to the outer polygon formed by the spine vertices in order to guarantee that the edges of the caterpillar are never crossed by other edges incident to the spine vertices. We now formally describe the algorithm.

\medskip\noindent{\bf Algorithm}~\textsc{Straight-line-Caterpillar}. Denote by 
$u_1, u_2, \dots, u_k$ the vertices of the spine of $S$. Also, for each spine 
vertex $u_i$ $(1 \leq i \leq k)$, let $v_{i1}, \dots, v_{in_i}$ be its leaves in 
$S$ (refer to the bottom image in Fig.~\ref{fi:spanning-caterpillar}). The 
algorithm temporarily adds to $S$ and $G$ some dummy vertices, which will be 
removed in the final drawing. Namely, for each $u_i$, it attaches to $u_i$ two 
dummy leaves, $s_i$ and $t_i$. Also, it adds a dummy spine vertex $u_{k+1}$ 
attached to $u_k$ and a dummy leaf $s_{k+1}$ to $u_{k+1}$ (see the top image in 
Fig.~\ref{fi:spanning-caterpillar}). Call $G'$ and $S'$ the new graph and the 
new caterpillar obtained by augmenting $G$ and $S$ with these dummy vertices. 

\begin{figure}[tb!]
  \centering
  	\subfigure[]{
    \includegraphics[height=0.5\textwidth] {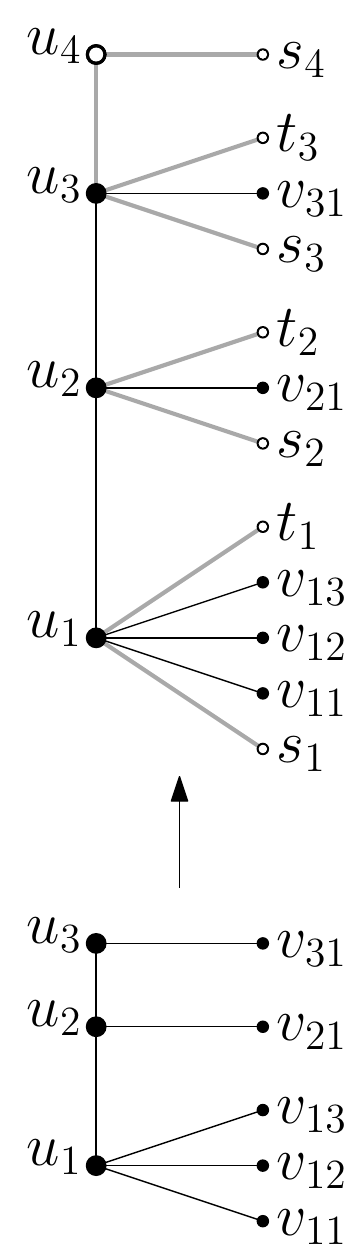}
    \label{fi:spanning-caterpillar}
	} 
    \hspace{0.05\textwidth}
  \subfigure[]{
    \includegraphics[height=0.5\textwidth] {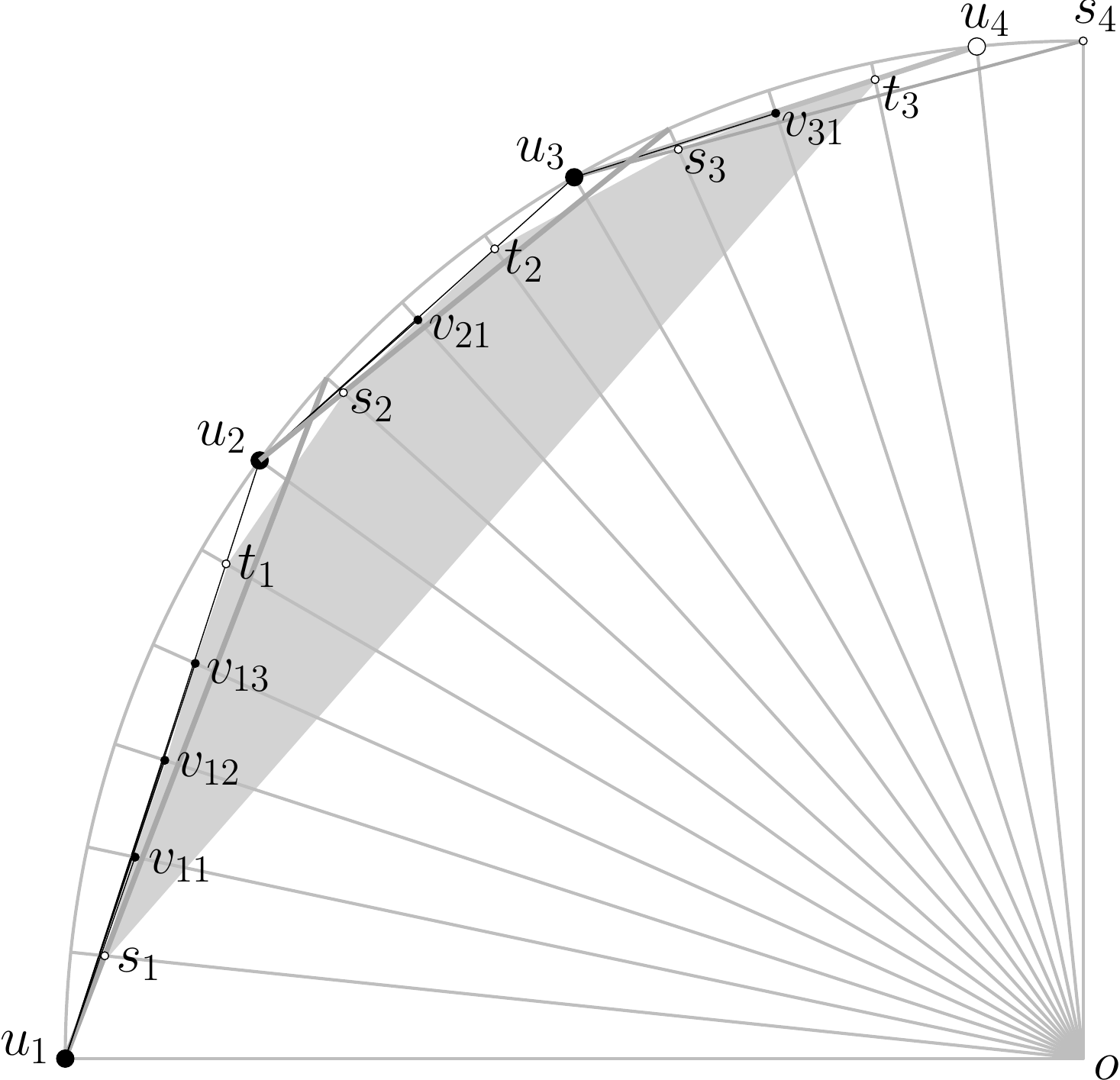}
    \label{fi:spanning-caterpillar-drawing}
	} 
    \hspace{0.05\textwidth}
  \subfigure[]{
    \includegraphics[height=0.5\textwidth] {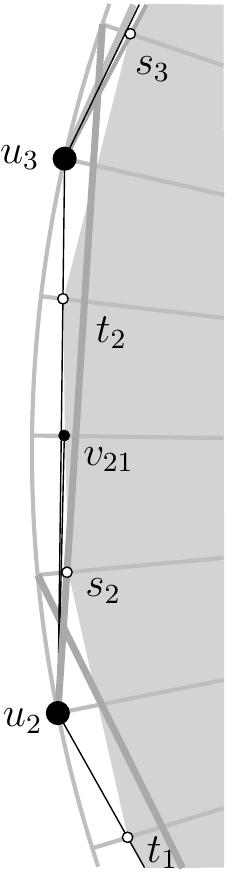}
    \label{fi:spanning-caterpillar-zoomed}
	} 
\caption{Illustration of Algorithm~\textsc{Straight-line-Caterpillar}: (a) a 
caterpillar $S$ and its augmented version $S'$; (b) a drawing of $S'$; edges of 
the graph connecting leaves of $S$ are drawn in the gray (convex) region; (c) 
enlarged detail of the picture (b).
}
\label{fi:caterpillar}
\end{figure}

The construction of a drawing $\Gamma'$ of $G'$ is illustrated in Fig.~\ref{fi:spanning-caterpillar-drawing}. 
Consider a quarter of circumference $C$ with center $o$ and radius $r$. Let $N$ 
be the total number of vertices of $G'$. Let $\{p_1, p_2, \dots, p_N\}$ be $N$ 
equally spaced points along $C$ in clockwise order, where $\overline{op_1}$ and 
$\overline{op_N}$ are a horizontal and a vertical segment, respectively. For 
each $1 \leq i \leq k$, consider the ordered list of vertices $L_i = \{u_i, s_i, 
v_{i1}, \dots v_{in_i}, t_i\}$, and let $L$ be the concatenation of all $L_i$. 
Also, append to $L$ the vertices $u_{k+1}$ and $s_{k+1}$, in this order. Clearly 
the number of vertices in $L$ equals $N$. For a vertex $v \in L$, denote by 
$j(v)$ the position of $v$ in $L$. Vertex $u_i$ is drawn at point $p_{j(u_i)}$ 
$(1 \leq i \leq k)$; also, vertices $u_{k+1}$ and $s_{k+1}$ are drawn at points 
$p_{N-1}$ and $p_{N}$, respectively. Each leaf $v$ of $S'$ will be suitably 
drawn along radius $\overline{op_{j(v)}}$ of $C$. More precisely, for any $i \in 
\{1, \dots, k$\}, let $a_i$ be the intersection point between segments 
$\overline{p_{j(u_i)}p_{j(s_{i+1})}}$ and $\overline{op_{j(s_i)}}$, and let 
$b_i$ be the intersection point between segments 
$\overline{p_{j(u_i)}p_{j(u_{i+1})}}$ and $\overline{op_{j(t_i)}}$. Vertices 
$s_i$ and $t_i$ are drawn at points $a_i$ and $b_i$, respectively. Also, let 
$A_i$ be the circular arc that is tangent to 
$\overline{p_{j(u_i)}p_{j(u_{i+1})}}$ at point $b_i$, and that passes through 
$a_i$; vertex $v_{ih}$ is drawn at the intersection point between $A_i$ and 
$\overline{op_{j(v_{ih})}}$ $(1 \leq h \leq n_i)$. 

Once all vertices of $G'$ are drawn, each edge of $G'$ is drawn in $\Gamma'$ as 
a straight-line segment between its end-vertices. Drawing $\Gamma$ is obtained 
from $\Gamma'$ by deleting all dummy vertices and their incident edges.

\begin{theorem}\label{th:spanning-caterpillar}
Let $G$ be a graph with $n$ vertices and $m$ edges, and let $S$ be a spanning 
caterpillar of $G$. There exists a straight-line \cdrawing $\Gamma$ of $\langle 
G,S \rangle$. Drawing $\Gamma$ can be computed in $O(n+m)$ time in the real RAM 
model\footnote{We also assume that basic trigonometric functions are executed in 
constant time.} and has $O(n^2)$ area. 
\end{theorem}

\begin{proof}
Compute $\Gamma$ by using Algorithm~\textsc{Straight-line-Caterpillar}. In the following we first 
prove that $\Gamma$ is a straight-line compatible drawing of $\langle G,S 
\rangle$, and then we analyze time complexity and area requirement. We adopt the 
same notation used in the description of the algorithm.\medskip

\noindent\textsc{Correctness.} We have to prove that in $\Gamma$ the edges of 
$S$ are never crossed. For a line $\ell$ denote by $s(\ell)$ its slope. Our 
construction places all spine vertices of $S'$ (and hence of $S$) in convex 
position. We claim that also all the leaves of $S'$ are in convex position. 
Indeed, this is clearly true for the subset of leaves of each $u_i$ $(1 \leq i 
\leq k)$, because this subset is drawn on a circular arc $A_i$; also, for any $i 
\in \{1, \dots, k-1\}$ consider the poly-line connecting the leaves of two 
consecutive spine vertices $u_i$, $u_{i+1}$, in the order they appear in $L$. 
This poly-line is convex if and only if the two segments incident to $s_{i+1}$ 
form an angle $\lambda$ smaller than $\pi$ on the side of the plane where the 
origin $o$ lies. In particular, let $\ell_1$ be the line through $t_i$ and 
$s_{i+1}$, and let $\ell_2$ be the line through $s_{i+1}$ and $v$, where $v$ 
coincides with $v_{{(i+1)}1}$ if such a vertex exists, while $v$ coincides with 
$t_{i+1}$ otherwise. Angle $\lambda < \pi$ if $s(\ell_1) > s(\ell_2)$. Denote by 
$c$ the intersection point between the chord $\overline{u_{i+1}u_{i+2}}$ and the 
radius $\overline{op_{j(v)}}$, and denote by $\ell_3$ the line through $s_{i+1}$ 
and $c$, we have $s(\ell_3) \geq s(\ell_2)$ by construction; then it suffices to 
show that $s(\ell_1) > s(\ell_3)$. For any fixed $N$, $s(\ell_3)$ is maximized 
when $c$ is as close as possible to $s_{i+1}$, i.e., when $n_{i+1}=0$; similarly 
$s(\ell_1)$ is minimized when $t_i$ is as close as possible to $s_{i+1}$, i.e., 
when $n_i=0$.       
For $n_i=n_{i+1}=0$ it can be verified by trigonometry that 
$\frac{s(\ell_1)}{s(\ell_3)} > 1$ (namely, this ratio tends to 1.23 when $N$ 
tends to infinity). Hence, the leaves of $S'$ except $s_{k+1}$ form a convex 
polygon $P$, which proves the claim.

Now, since by construction the edges of $S$ are all outside $P$ in $\Gamma$, 
these edges cannot be crossed by edges of $G$ connecting two leaves of $S$. It 
is also immediate to see that an edge of $S$ cannot be crossed by another edge 
of $S$. It remains to prove that an edge of $S$ cannot be crossed by an edge of 
$G$ connecting either two non-consecutive spine vertices or a leaf of $S$ to a 
spine vertex of $S$. 

There are two kinds of edges in $S$. Edges $(u_i,u_{i+1})$, connecting two 
consecutive spine vertices, and edges $(u_i,v_{ih})$, connecting a spine vertex 
to its leaves. Since by construction $\Gamma$ is totally drawn inside the closed 
polygon formed by the spine vertices of $S$ (recall that $u_{k+1}$ and $s_{k+1}$ 
are dummy vertices, and then they do not belong to $\Gamma$), edges 
$(u_i,u_{i+1})$ cannot be crossed in $\Gamma$.    
Now, consider an edge $(u_i,v_{ih})$. Obviously, it cannot be crossed by an 
edge $(u_i, u_j)$, because two adjacent edges cannot cross in a straight-line 
drawing; yet, it cannot be crossed by an edge $(u_j, u_z)$ where $j < z$ and 
$j,z \neq i$. Indeed, if $i < j$ or $i > z$ then there is a line $\ell$ through 
$o$ such that $(u_i,v_{ih})$ completely lies in one of the two half planes 
determined by $\ell$ and $(u_j,u_z)$ completely lies in the other half plane; 
also, if $j < i < z$, then edge $(u_i,v_{ih})$ completely lies in the open 
region delimited by $(u_j,u_z)$ and $C$. We finally show that $(u_i,v_{ih})$ 
cannot be crossed by any edge $(u_j,v_{df})$, where $j \neq i$ and $v_{df} \neq 
v_{ih}$. Indeed, if $d > i$ and $j > i$, or $d < i$ and $j < i$, then 
$(u_i,v_{ih})$ and $(u_j,v_{df})$ are completely separated by a line through 
$o$. 
Also, if $d < i$ and $j > i$ (or $d > i$ and $j < i$), edge $(u_i,v_{ih})$ lies 
completely outside the triangle with vertices $o, u_j, v_{df}$, thus it cannot 
cross edge $(u_j,v_{df})$.      
     
\medskip

\noindent\textsc{Time and area requirement.} It is immediate to see that the 
construction of $\Gamma'$ (and then of $\Gamma$) can be executed in linear time, 
in the real RAM model. It remains to prove that the area of $\Gamma$ is 
$O(n^2)$. Assume that $o$ coincides with the origin of a Cartesian coordinate 
system, so that $p_1$ has coordinates $(-r,0)$ and $p_N$ has coordinates 
$(0,r)$. We need to estimate the minimum distance $d_{\min}$ between any two 
points of $\Gamma$. According to our construction, the vertex at position $i$ in 
$L$ is drawn on a point $q_i$ along radius $\overline{op_i}$, and $d_{\min}$ 
corresponds to the minimum distance between any two points $q_i$ and $q_{i+1}$ 
$(1 \leq i \leq N-1)$. Also, denote by $p'_i$ the intersection point between 
radius $\overline{op_i}$ and the chord $\overline{p_1p_N}$, point $q_i$ is 
in-between $p_i$ and $p'_i$ along $\overline{op_i}$; this implies that 
$d'_{\min} \leq d_{\min}$, where $d'_{\min}$ is the minimum distance between any 
two points $p'_i$ and $p'_{i+1}$.
Now, it is immediate to observe that $d'_{\min}$ equals the length of segment 
$\overline{p'_{\lceil N/2 \rceil}p'_{\lceil N/2 \rceil +1}}$. 
Let $p'$ be the middle point of chord $\overline{p_1p_N}$; $p'$ has coordinates 
$(\frac{-r}{2}, \frac{r}{2})$. Point $p'$ coincides with $p'_{\lceil N/2 
\rceil}$ if $N$ is odd, while it is equidistant to $p'_{\lceil N/2 \rceil}$ and 
$p'_{\lceil N/2 \rceil +1}$ if $N$ is even. Hence, denoted by $d'$ the distance 
between $p'$ and $p'_{\lceil N/2 \rceil +1}$, we have that $d' \leq d'_{\min} 
\leq d_{\min}$. Now, let $\alpha=\frac{\pi}{2(N-1)}$ be the angle at $o$ formed 
by any two radii $\overline{op_i}$, $\overline{op_{i+1}}$, and let $\beta$ be 
the angle at $o$ formed by $\overline{op'}$ and $\overline{op'_{\lceil N/2 
\rceil +1}}$. Clearly $\beta = \alpha$ if $N$ is odd, while $\beta = 
\frac{\alpha}{2}$ if $N$ is even. 
Also, since $\overline{op'}$ forms a right angle with $\overline{p'p'_{\lceil 
N/2 \rceil +1}}$, and since the length of $\overline{op'}$ is 
$\frac{r}{\sqrt{2}}$, we have that $d' = \frac{r \tan\beta}{\sqrt{2}}$.   
Hence, if we let $d'=1$ (which guarantees that $d_{\min} \geq 1$), we obtain $r 
= \frac{\sqrt{2}}{\tan\beta}$.
Since $\tan\beta > \beta$ for $\beta \in (0,\frac{\pi}{2})$, then $r < 
\frac{\sqrt{2}}{\beta}$, which implies $r = O(N)$, because $\beta = 
\theta(\frac{1}{N})$ . Thus, the area of $\Gamma$ is $O(N^2)=O(n^2)$.
\end{proof}


The following lemmas show that, unfortunately, Theorem~\ref{th:spanning-spider} and
Theorem~\ref{th:spanning-caterpillar} cannot be extended to every spanning tree
$S$, that is, there exist pairs $\langle G,S \rangle$ that do not admit straight-line \cdrawings, even if $S$ is a ternary or a binary tree.

\begin{lemma}\label{le:bad-trees-1}
Let $G$ be the complete graph on $13$ vertices and let $S$ be a complete rooted
ternary spanning tree of $G$. There is no straight-line \cdrawing of $\langle G,S \rangle$.
\end{lemma}

\begin{figure}[tb!]
  \centering
  	\subfigure[]{
    \includegraphics[width=0.35\textwidth] {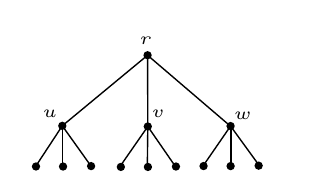}
    \label{fi:tree-13}
	}
    \subfigure[]{
    \includegraphics[width=0.35\textwidth] {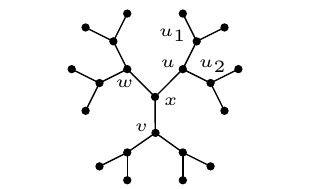}
    \label{fi:tree-22}
	}
\caption{
(a) The complete rooted ternary tree with $13$ vertices in the statement of 
Lemma~\ref{le:bad-trees-1}.
(b) The complete binary tree with $22$ vertices in the statement 
of Lemma~\ref{le:bad-trees-2}.}
\end{figure}

\begin{proof}
We show by case analysis that there is no straight-line \cdrawing of $\langle G,S \rangle$. 
Let $r$ be the root of $S$ (see Fig.~\ref{fi:tree-13}). Note that $r$ is the only vertex of $S$ with
degree $3$. Let $u$, $v$, $w$ be the three neighbors of $r$ in $S$. Two are the
cases, either one of $u$, $v$, $w$ (say $u$) lies inside triangle
$\triangle(r,v,w)$ (Case $1$, see Fig.~\ref{fi:13-case1}), or $r$
lies inside triangle $\triangle(u,v,w)$ (Case $2$, see Fig.~\ref{fi:13-case2}).

In Case $1$, consider a child $u_1$ of $u$. In any straight-line \cdrawing of
$\langle G,S \rangle$, $u_1$ is placed in such a way that $u$ lies inside
either triangle $\triangle(u_1,r,w)$ or triangle $\triangle(u_1,r,v)$, assume
the former (see Fig.~\ref{fi:13-case1-u1}). Then, consider another child $u_2$
of $u$; in order for edge $(u,u_2)$ not to cross any edge, also $u_2$ has to lie
inside $\triangle(u_1,r,w)$, in such a way that both $u$ and $u_1$ lie inside
triangle $\triangle(u_2,r,v)$. This implies that $u$ lies inside
$\triangle(u_1,r,u_2)$ (see Fig.~\ref{fi:13-case1-u2}), together with its
last child $u_3$. However, $u_3$ cannot be placed in any of the three triangles
in which $\triangle(u_1,r,u_2)$ is partitioned by the edges (of $S$) connecting
$u$ to $u_1$, to $r$, and to $u_2$, respectively, without introducing any
crossing involving edges of $S$. This concludes the analysis of Case $1$.

In Case $2$, note that any child $u_1$ of $u$ cannot be drawn inside
$\triangle(u,v,w)$, as otherwise one of the edges of $S$ incident to $r$ would
be crossed by either $(u,u_1) \in S$, $(u_1,v) \in G$ or $(u_1,w) \in G$. We
further distinguish two cases, based on whether $u$ and $r$ lie inside triangle
$\triangle(u_1,v,w)$ (Case $2.1$, see Fig.~\ref{fi:13-case21}), or $r$ and
one of $v$ and $w$ (say $w$) lies inside triangle $\triangle(u_1,u,v)$ (Case
$2.2$, see Fig.~\ref{fi:13-case22}).
In Case $2.1$, consider another child $u_2$ of $u$. Note that, $u_2$ has
to lie inside $\triangle(u_1,v,w)$, due to edge $(u,u_2) \in S$. However, $u_2$
cannot be placed in any of the three regions in which $\triangle(u_1,v,w)$ is
partitioned by paths composed of edges of $S$ connecting $r$ to $u_1$, to $v$,
and to $w$, respectively. To conclude the proof, note that, if Case $2.2$ holds
for the children of vertex $u$, then Case $2.1$ must hold for the children of
vertex $w$, as all of them must lie inside $\triangle(u_1,u,v)$.
\end{proof}

\begin{figure}[tb!]
  \centering
  \subfigure[]{
    \includegraphics[width=0.25\textwidth] {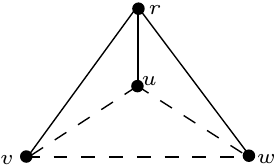}
    \label{fi:13-case1}
	} 
  \subfigure[]{
    \includegraphics[width=0.25\textwidth] {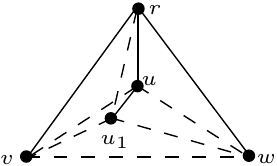}
    \label{fi:13-case1-u1}
	}
  \subfigure[]{
    \includegraphics[width=0.25\textwidth] {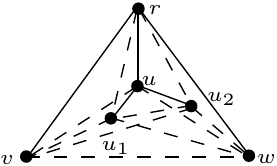}
    \label{fi:13-case1-u2}
	}
	\\
  \subfigure[]{
    \includegraphics[width=0.25\textwidth] {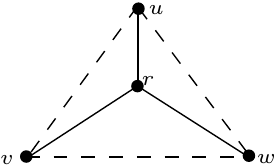}
    \label{fi:13-case2}
	} 
    \subfigure[]{
    \includegraphics[width=0.25\textwidth] {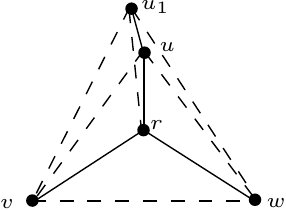}
    \label{fi:13-case21}
	}
    \subfigure[]{
    \includegraphics[width=0.25\textwidth] {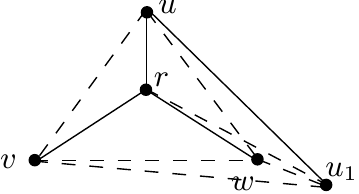}
    \label{fi:13-case22}
	}
\caption{
Illustration for Lemma~\ref{le:bad-trees-1}:
(a) Case $1$ in the proof; $u$ lies inside 
$\triangle(r,v,w)$. (b) Placement of $u_1$. (c) Placement of $u_2$.
(d) Case $2$ in the proof; $r$ lies inside 
$\triangle(u,v,w)$. (e) Case $2.1$. (f) Case $2.2$. 
}
\end{figure}

\begin{lemma}\label{le:bad-trees-2}
Let $G$ be the complete graph on $22$ vertices and let $S$ be a complete
unrooted binary spanning tree of $G$. There is no straight-line
\cdrawing of $\langle G,S \rangle$. 
\end{lemma}

\begin{proof}
First, we claim the following property ($P1$): Let $x$ be a vertex of $G$
such that the neighbors $u,v,w$ of $x$ in $S$ are not leaves of $S$. Then, in
any straight-line \cdrawing of $\langle G,S \rangle$, vertex $x$ lies outside
triangle $\triangle(u,v,w)$.
Observe that property $P1$ directly descends from Case $2$ of the proof of
Lemma~\ref{le:bad-trees-1}, where $x$ plays the role of $r$. Indeed, in that
proof, only two of the children of $u$ (and of $v$ and $w$, simmetrically) were
used in the argument.

Now, consider the only vertex $x$ of $G$ such that each of the tree subtrees of
$S$ rooted at $x$ contains seven vertices (see Fig.~\ref{fi:tree-22}). By $P1$, vertex $x$ lies
outside the triangle $\triangle(u,v,w)$ composed of its neighbors $u$, $v$,
$w$, which implies that one of $u$, $v$, $w$ (say $u$) lies inside triangle
$\triangle(x,v,w)$. As in the proof of Case $1$ of Lemma~\ref{le:bad-trees-1},
with $x$ playing the role of $r$, we observe that in any
straight-line \cdrawing of $\langle G,S \rangle$ in which $u$ lies inside
$\triangle(x,v,w)$, the two neighbors $u_1$ and $u_2$ of $u$ are placed in such
a way that $u$ lies inside $\triangle(u_1,x,u_2)$. 
While in Lemma~\ref{le:bad-trees-1} we used the presence of a fourth neighbor
(a third child) of $u$ to prove the statement, here we can apply $P1$, as $u_1$,
$x$, $u_2$ are not leaves of $S$.
\end{proof}

In light of Lemmas~\ref{le:bad-trees-1} and~\ref{le:bad-trees-2}, it is 
natural to ask whether there are specific subfamilies of spanning trees $S$ 
(other than paths, spiders, and caterpillars) such that a straight-line 
\cdrawing of $\langle G, S \rangle$ always exists. The following algorithm gives a 
positive answer to this question: it computes a straight-line \cdrawing when $S$ 
is a BFS-tree of $G$. Theorem~\ref{th:good-trees} proves the algorithm 
correctness, its time complexity, and its area requirement.

The idea of the algorithm is to exploit the properties of BFS-trees. Namely, two consecutive levels of a BFS-tree induce a subgraph that is a forest of caterpillars and there is no edge spanning more than one level. The algorithm is based on a recursive technique that uses, in each recursive step, an argument similar to that used for caterpillars, although the geometric construction is different. The final result is a drawing composed of a set of nested convex polygons, one for each level of the BFS-tree. Edges connecting vertices of two consecutive levels are drawn between the two corresponding polygons and are not crossed by other edges. All the remaining edges are allowed to cross.

\begin{figure}[tb!]
  \centering
  	\subfigure[]{
    \includegraphics[height=0.55\textwidth] {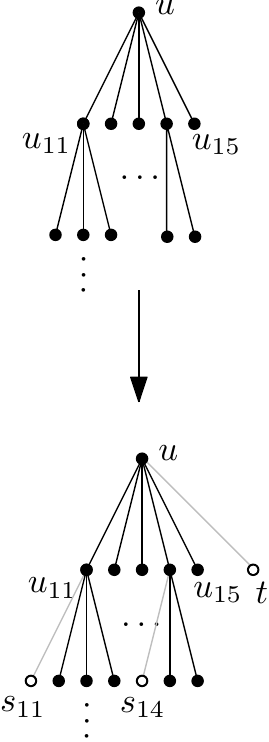}
    \label{fi:proper-level-tree}
    }
    \hspace{0.07\textwidth}
    \subfigure[]{
    \includegraphics[height=0.55\textwidth] {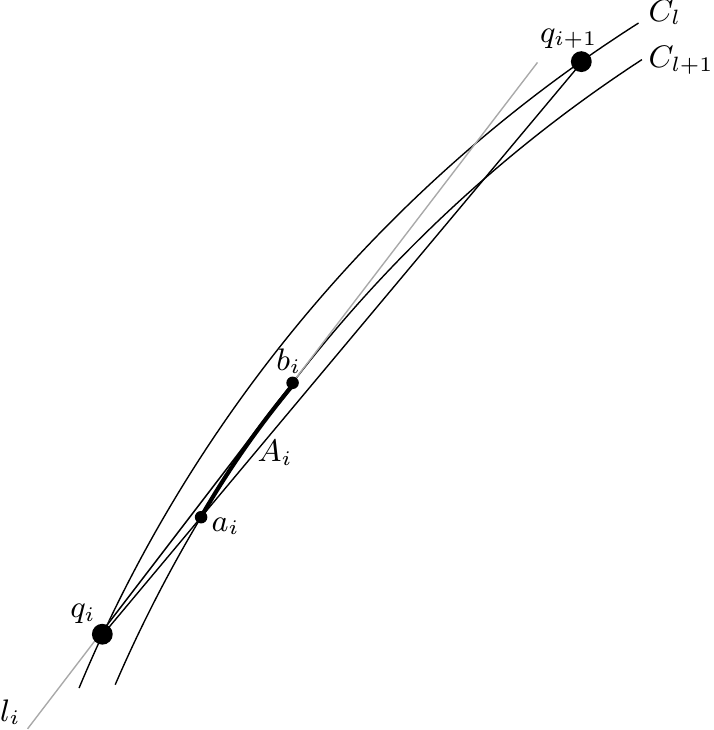}
    \label{fi:proper-level-tree-zoom}
	}     
    \subfigure[]{
    \includegraphics[height=0.5\textwidth] {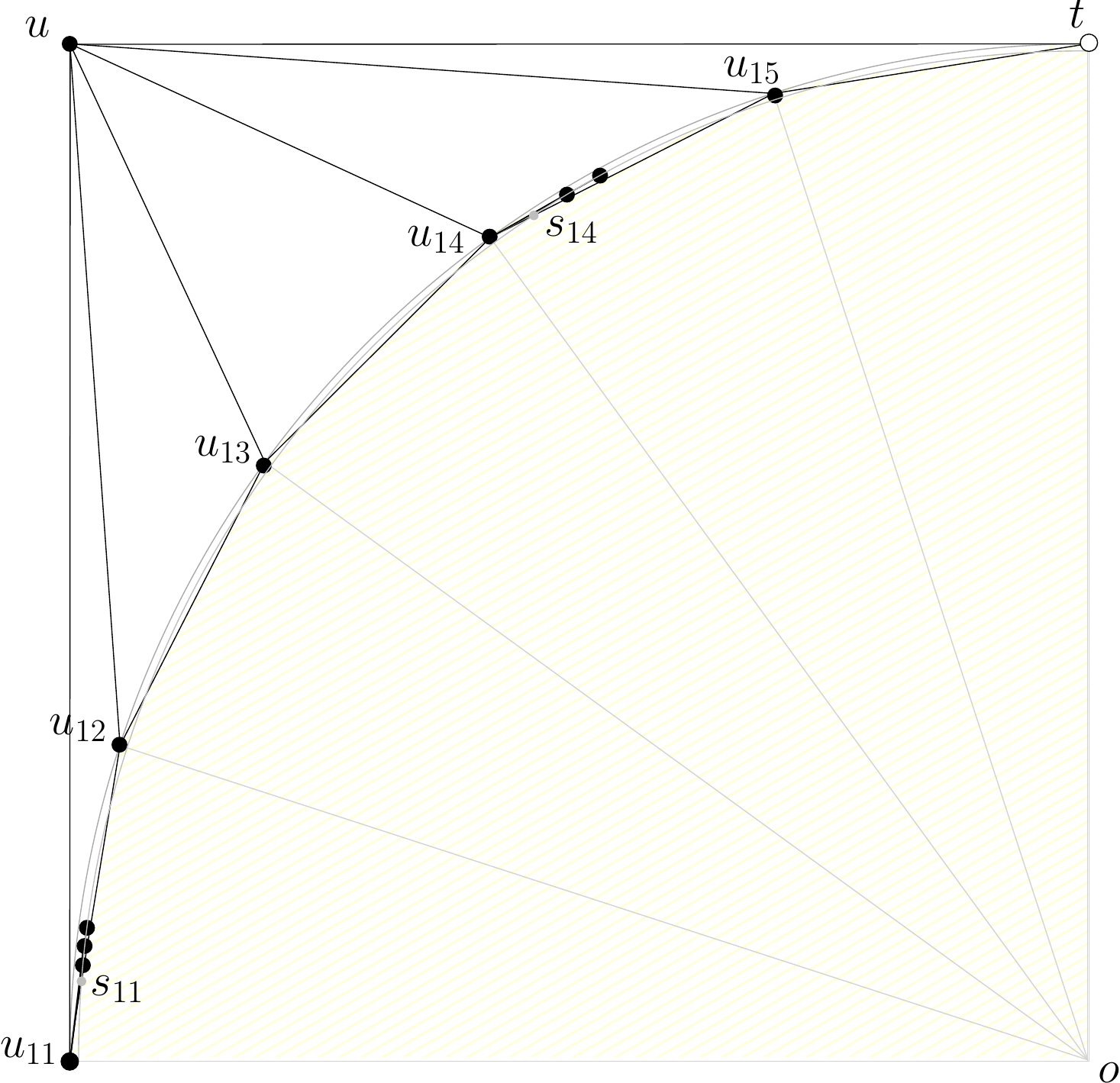}
    \label{fi:proper-level-tree-drawing}
	} 
\caption{Illustration of Algorithm~\textsc{Straight-Line-BFS-Tree}: (a) a 
BFS-tree $S$ and its enhanced version $S'$. (b) Inductive construction of a 
portion of the drawing. (c) Sketch of the final drawing $\Gamma$.}
\label{fi:good-tree}
\end{figure}

\medskip\noindent{\bf Algorithm}~\textsc{Straight-Line-BFS-Tree}. 
Refer to Fig.~\ref{fi:good-tree} for an illustration of the 
algorithm.
Let $u$ be the root of $S$ (which is at level $0$) and let $u_{l1}$, $\dots$, 
$u_{lk_l}$ be the vertices at level $l \in \{1, \dots, d\}$, where $d$ is the 
depth of $S$.   
The algorithm temporarily adds to $S$ and $G$ some dummy vertices, which will 
be removed in the final drawing. 
Namely, for each $u_{li}$, $1 \leq l \leq d-1$ and $1 \leq i \leq k_l$, it 
attaches to $u_{li}$ one more (leftmost) child $s_{li}$. Also, it attaches to 
root $u$ a dummy (rightmost) child $t$. 
Denote by $G'$ and $S'$ the new graph and the new tree, respectively. Notice 
that $S'$ is still a BFS-tree of $G'$. The algorithm iteratively computes a 
drawing $\Gamma'$ of $G'$. For $l = 1, \dots, d$, the algorithm defines a 
circumference $C_l$ with center $o=(0,0)$ and radius $r_l < r_{l-1}$ ($C_1, 
\dots, C_d$ are concentric). The vertices of level $l$ are drawn on the quarter 
of $C_l$ going from point $(-r_l,0)$ to point $(0,r_l)$ clockwise. 
 

Let $\{u_{11}$, $\dots$, $u_{1k_1}$, $t\}$ be the ordered list of the children 
of root $u$ and let $\{p_{11}$,$\dots$, $p_{1k_{1}}$, $p_t\}$ be $k_{1}+1$ 
equally spaced points along $C_1$ in clockwise order, where $\overline{op_{11}}$ 
and $\overline{op_{t}}$ are a horizontal and a vertical segment, respectively. 
Vertex $u_{1j}$ is drawn on $p_{1j}$ ($1 \leq j \leq k_{1}$) and vertex $t$ is 
drawn on $p_t$. Also, $u$ is drawn on point $(-r_1,r_1)$.

Assume now that all vertices $u_{l1},\dots, u_{lk_{l}}$ of level $l$ have been 
drawn ($1 \leq l \leq d-1$) in this order on the sequence of points 
$\{q_1,\dots, q_{k_l}\}$, along $C_l$. The algorithm draws the vertices of level 
$l+1$ as follows. Let $\overline{q_{i}q_{i+1}}$ be the chords of $C_{l}$, for $1 
\leq i \leq k_{l}-1$, and let $c_{l}$ be the shortest of these chords. The 
radius $r_{l+1}$ of $C_{l+1}$ is chosen arbitrarily in such a way that $C_{l+1}$ 
intersects $c_{l}$ in two points and $r_{l+1} < r_l$. This implies that 
$C_{l+1}$ also intersects every chord 
$\overline{q_{i}q_{i+1}}$ in two points. For $1 \leq i \leq k_{l}$, denote by 
$L(u_{l i})$ = $\{v_{1}$, $\dots$, $v_{n_{l i}}\}$ the ordered list of children 
of $u_{l i}$ in $G'$.  Also, let $a_i$ be the intersection point between 
$\overline{q_{i}q_{i+1}}$ and $C_{l+1}$ that is closest to $q_i$, and let 
$\ell_i$ be the line through $q_i$ tangent to $C_{l+1}$; denote by $b_i$ the 
tangent point between $\ell_i$ and $C_{l+1}$. 
Let $A_{l+1}$ be the arc of $C_{l+1}$ between $a_i$ and $b_i$, and let $\{p_0, 
p_{1},\dots, p_{n_{l i}}\}$ be $n_{l i}+1$ equally spaced points along $A_{l+1}$ 
in clockwise order. For $v \in L(u_{l i})$, denote by $j(v)$ the position of $v$ 
in $L(u_{l i})$. Vertex $v_{j}$ is drawn on $p_{j(v_{j})}$ ($1 \leq j \leq n_{l 
i}$) and vertex $s_{l i}$ is drawn on $p_{0}$. 

Once all vertices of $G'$ are drawn each edge of $G'$ is drawn in $\Gamma'$ as 
a straight-line segment between its end-vertices. Drawing $\Gamma$ is obtained 
from $\Gamma'$ by deleting all dummy vertices and their incident edges.  

\begin{theorem}\label{th:good-trees}
Let $G$ be a graph with $n$ vertices and $m$ edges, and let $S$ be a BFS-tree 
of $G$. 
There exists a straight-line \cdrawing $\Gamma$ of $\langle G, S \rangle$. 
Drawing $\Gamma$ can be computed in $O(n+m)$ time in the real RAM model. 
\end{theorem}

\begin{proof}
The algorithm that constructs $\Gamma$ is 
Algorithm~\textsc{Straight-Line-BFS-Tree}. In the following we first prove that 
$\Gamma$ is a straight-line compatible drawing of $\langle G,S \rangle$, and 
then we analyze the time complexity. We adopt the same notation used in the 
description of the algorithm.\medskip

\noindent\textsc{Correctness.} 
We have to prove that in $\Gamma$ the edges of $S$ are never crossed.  
Observe that, since $S$ is a BFS-tree of $G$, there cannot be edges spanning 
more than two consecutive levels of $S$. Following its description, we prove the 
correctness of the algorithm by induction on $l$. 

In the base case consider the levels 0 and 1, also, consider the convex polygon 
$P_1$, whose vertices are the points $\{p_1$,$\dots$, $p_{n_{u}}\}$. 
All the edges connecting two children of root $u$ are drawn inside $P_1$, 
and do not cross the edges connecting the root to its children, which are drawn 
outside $P_1$.

Assume by induction that all the edges in $S$ connecting two vertices of two 
levels $1 \leq l',l'+1 \leq l$ are not crossed in $\Gamma$. We prove that all 
the edges in $S$ connecting two vertices of levels $l$, $l+1$ are not crossed in 
$\Gamma$. Consider any vertex $u_{li}$ $(1 \leq i \leq k_l)$; 
observe that all its children are drawn inside the open plane region $R$ 
defined by the arc of $C_l$ that goes clockwise from $q_i$ to $q_{i+1}$ (where 
$u_{li}$ and $u_{l i+1}$ are drawn) 
and the chord $\overline{q_iq_{i+1}}$. 
By construction, $R$ is never intersected by an edge connecting two vertices 
drawn in a step $l' \leq l$. 
Consider the convex polygon $P_{l+1}$, defined by the points where the vertices 
of level $l+1$ are drawn. 
All the edges connecting two vertices of level $l+1$ are drawn inside $P_{l+1}$, 
and do not cross the edges connecting vertices of level $l$ to their children, 
which are drawn outside $P_{l+1}$. 
It remains to prove that every edge $e' \notin S$, connecting a vertex of level 
$l$ to a vertex of level $l+1$, 
does not cross any edge $e \in S$, connecting a vertex of level $l$ to a vertex 
of level $l+1$. 
In particular, let $e = (u_{li}, u_{l+1 j}) \in S$ ($1 \leq i \leq k_l$ and $1 \leq j \leq k_{l+1}$) 
and $e'=(u_{lz}, u_{l+1 f}) \notin S$ ($1 \leq z \leq k_l$ and $1 \leq f \leq k_{l+1}$). 
Assume that $u_{l+1 f}$ is not a child of $u_{li}$ in $S$. 
If $i < z$ and $j < f$ or $i > z$ and $j > f$, then there is a line $\ell$ 
through $o$ such that $e$ completely lies in one of the two half planes 
determined by $\ell$ and $e'$ completely lies in the other half plane. 
If $i < z$ and $j > f$ or $i > z$ and $j < f$, consider the line $\ell$ 
containing the straight-line segment $\overline{u_{li}u_{lz}}$, then $e$ 
completely lies in one of the two half planes determined by $\ell$ (the one 
containing $u_{{l+1}j}$), and $e'$ completely lies in the other half plane; 
indeed, $r_{l+1}$ has been chosen so that it intersects $c_l$ (the 
minimum-length chord of $C_l$). 
Finally, suppose $u_{l+1 f}$ is a child of $u_{li}$ in $S$; by construction, 
any edge that connects $u_{lz}$ to a vertex of level $l+1$ (which is not a child 
of $u_{lz}$), including $e'$, must cross the circumference $C_{l+1}$ exactly 
once (near the point where $u_{lz}$ is placed).

\medskip
\noindent\textsc{Time requirement.} 
At each inductive step, the technique performs a number of operations 
proportional to $k_l + k_{l+1}$. Indeed, $c_l$ is chosen by looking only at the 
chords between consecutive points on $C_l$. Hence, the overall time complexity 
is $O(\sum_{l=0}^{d-1} (k_l + k_{l+1}) + m)$ = $O(\sum_{l=0}^{d-1} k_l) + 
O(\sum_{l=0}^{d-1} k_{l+1}) + O(m)$ = $O(n + m)$. 

Observe that the \cdrawing computed by 
Algorithm~\textsc{Straight-Line-BFS-Tree} may require area $\Omega(2^n)$. 
Indeed, let $\mathcal{L}(C_1)$ be the length of $C_1$; the children of the 
root 
$u$ are drawn along an arc $A_1$ of $C_1$ whose length is $\mathcal{L}(A_1) < 
\mathcal{L}(C_1)/2$. Inductively, the children of any vertex of level $l-1$ are 
drawn along an arc $A_l$ of $C_l$ whose length is $\mathcal{L}(A_l) < 
\mathcal{L}(A_{l-1})/2 < \mathcal{L}(C_1)/2^l$. 
Hence, the children of any vertex of level $d-1$ are drawn along an arc of 
circumference $A_d$ whose length is $\mathcal{L}(A_d) < \mathcal{L}(C_1)/2^d$. 
It follows that the minimum distance between any two points in $\Gamma$ is 
$d_{\min} = o(\mathcal{L}(C_1)/2^d)$. Consider the case $d \in O(n)$, and 
impose 
$d_{\min} = 1$, it follows that $\mathcal{L}(C_1) \in \Omega(2^n)$, which 
implies that the area of $\Gamma$ is $\Omega(2^n)$.
\end{proof}

It is worth observing that any graph $G$ admits a 
BFS-tree rooted at an arbitrarily chosen vertex $r$ of $G$.
Indeed, it corresponds to the spanning tree computed with a breadth-first-search starting from $r$. 
Thus, each graph admits a straight-line drawing $\Gamma$ such that one of its 
spanning trees $S$ is never crossed in $\Gamma$.

\subsection{Spanning Triconnected Subgraphs}\label{sse:straight-line-triconnected}

Here we focus on the case in which $S$ is a triconnected spanning subgraph of $G$. 
Clearly, since every tree can be augmented with edges to become a 
triconnected graph, Lemmas~\ref{le:bad-trees-1} and~\ref{le:bad-trees-2} imply 
that, if $S$ is a triconnected graph, a straight-line \cdrawing of $\langle G,S 
\rangle$ may not exist. The following theorem characterizes those instances for which 
such a drawing exists.

\begin{theorem}\label{th:triconnected-decision}
Let $G(V,E)$ be a graph, $S(V,W)$ be a planar triconnected spanning subgraph of 
$G$, and $\cal E$ be the unique planar (combinatorial) embedding of $S$ (up to 
a flip). A straight-line \cdrawing $\Gamma$ of $\langle G,S \rangle$ exists if 
and only~if: 
\begin{enumerate}[(1)]
\item Each edge $e \in E\setminus W$ connects two vertices belonging 
to the same face of~$\cal E$. 
\item There exists a face $f$ of $\cal E$ containing three vertices such that any 
pair $u,v$ of them does not separate in the circular order of $f$ the 
end-vertices $x,y \in f$ of any other edge in $E\setminus W$. 
\end{enumerate}
\end{theorem}
 
\noindent

\begin{proof}
  Suppose that $v_1,v_2$, and $v_3$ are three vertices of a face $f$ satisfying 
  Condition~2 (see Fig.~\ref{fi:trico-a}). Consider the graphs $G^*(V,E \cup 
\Delta)$ and $S^*(V, W\cup 
  \Delta)$, where $\Delta = \{(v_1,v_2),(v_2,v_3),(v_1,v_3)\}$. Observe that, 
due to Condition~2,  $S^*$ is a triconnected planar spanning subgraph of $G^*$ 
and that $\Delta$ forms an empty triangular face in the unique planar embedding 
of $S^*$. Produce a convex straight-line drawing $\Gamma^*$ of $S^*$ 
  with $\Delta$ as the external face (for example, using the algorithm 
in~\cite{br-scdpg-06}). A straight-line \cdrawing $\Gamma$ of $\langle G,S 
\rangle$ 
can be obtained from $\Gamma^*$ by removing the edges in $\Delta$ and by adding 
  the edges of $E \setminus W$. Observe that by Condition~1 and 
  by the convexity of the faces of $\Gamma^*$, the edges of $E \setminus W$ do 
  not intersect edges of $S$.
  
  Conversely, suppose that $\langle G,S \rangle$ admits a straight-line 
\cdrawing $\Gamma$. 
  Condition~1 is trivially satisfied. Regarding Condition~2, consider the 
  circular sequence $v_1,v_2,\dots,v_k$ of vertices of $V$ on the convex hull 
of $\Gamma$ (see Fig.~\ref{fi:trico-b}). 
  Observe that $k \ge 3$ and that any triple of such vertices satisfies 
Condition~2.  
 \end{proof}

\begin{figure}[tb!]
  \centering
  	\subfigure[]{
    \includegraphics[height=0.25\textwidth] {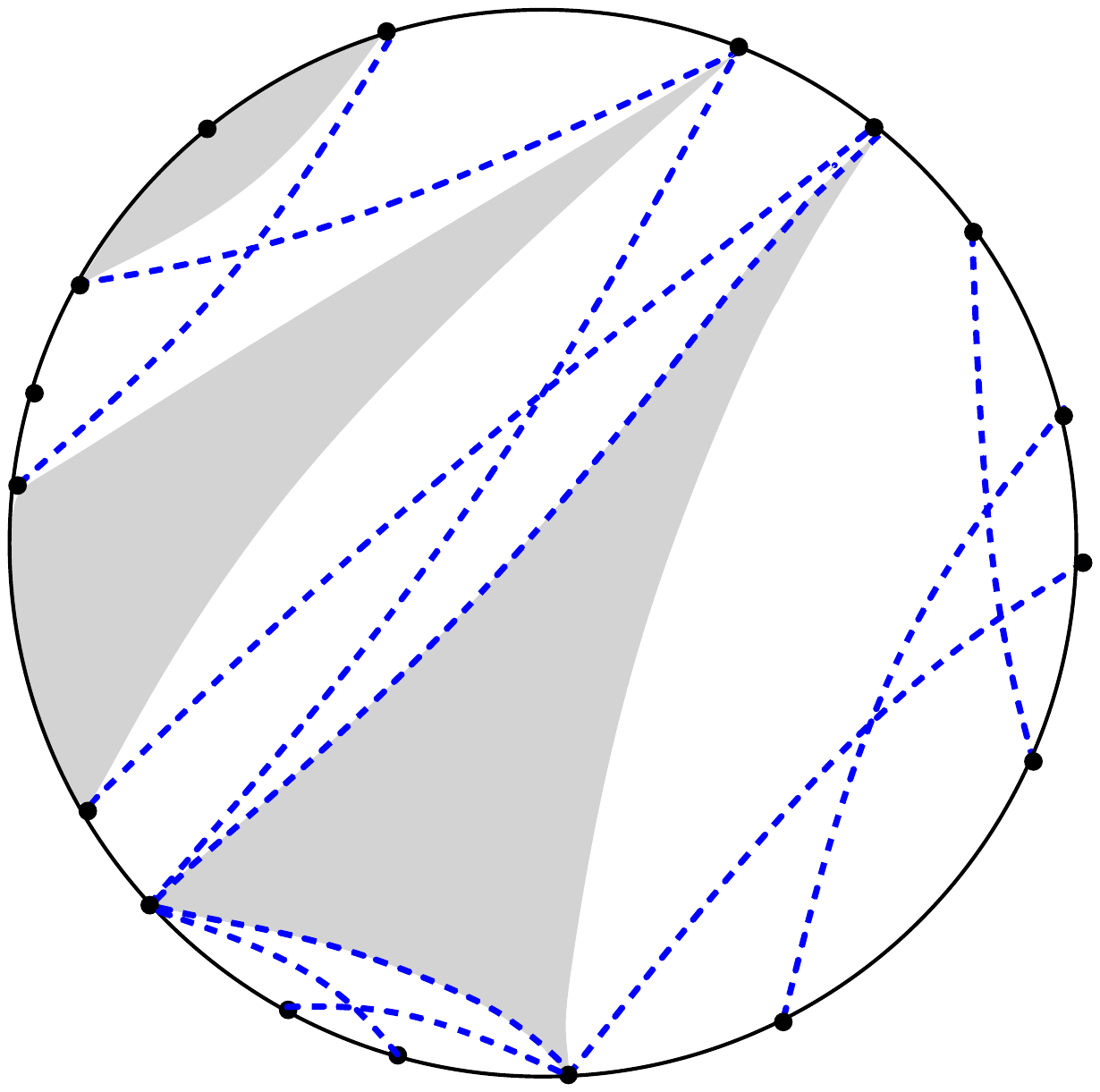}
    \label{fi:trico-a}
	} 
    \hspace{0.1\textwidth}
  \subfigure[]{
    \includegraphics[height=0.25\textwidth] {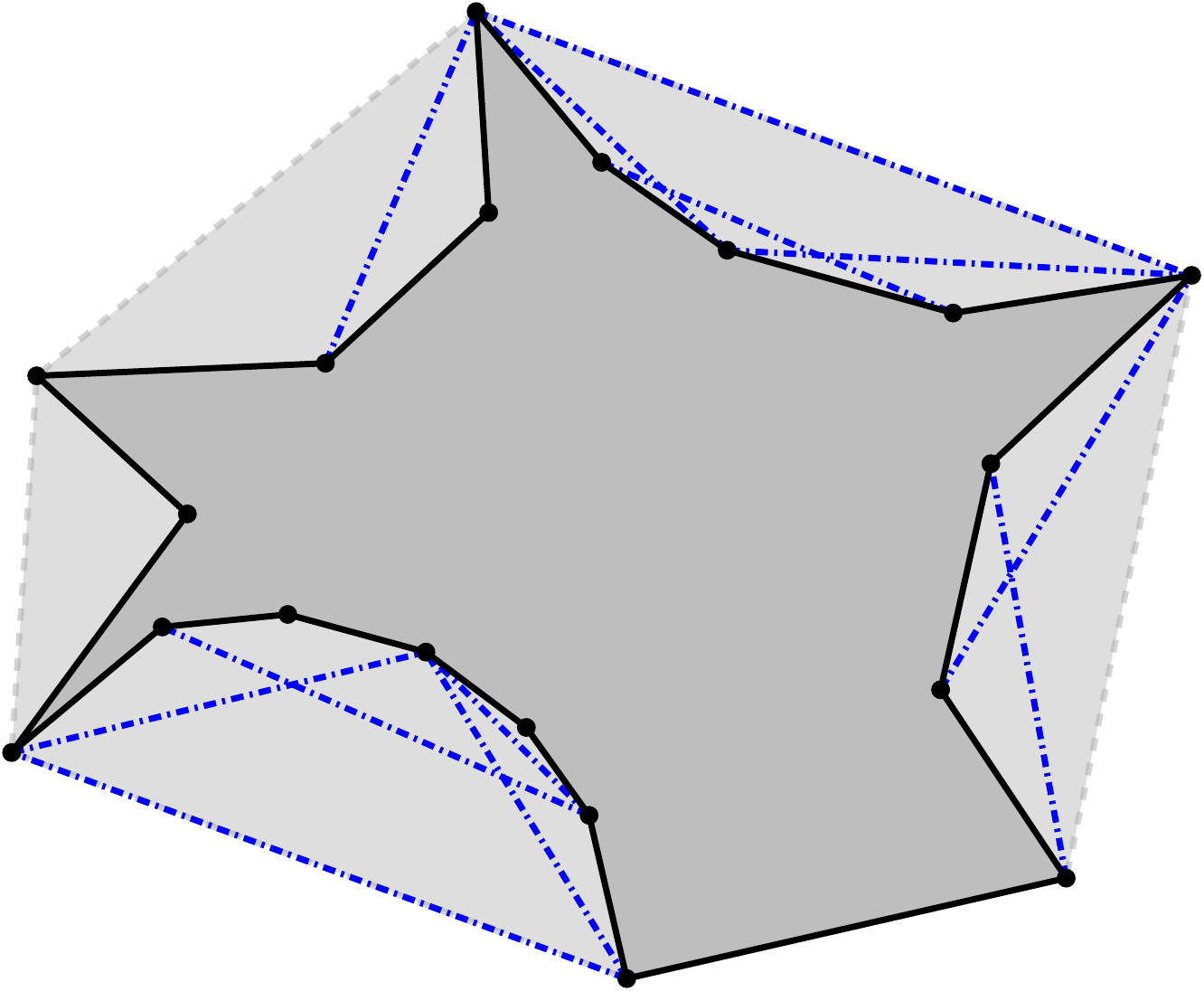}
    \label{fi:trico-b}
	} 
\caption{
(a) A face of $\cal E$, where edges of $E \setminus W$ are drawn as dashed curves. 
Shaded triangles identify three triplets of vertices among those satisfying 
Condition~2. (b) A straight-line \cdrawing of $\langle G, S \rangle$ used to 
show the necessity of Condition~2.
}
\label{fi:trico-theorem}
\end{figure}

In the following we describe an algorithm to test in polynomial time whether 
the conditions of Theorem~\ref{th:triconnected-decision} are satisfied by a 
pair $\langle G,S \rangle$ in which $S$ is triconnected and spanning.

\begin{figure}[tb!]
  \centering
  	\subfigure[]{
    \includegraphics[height=0.25\textwidth] {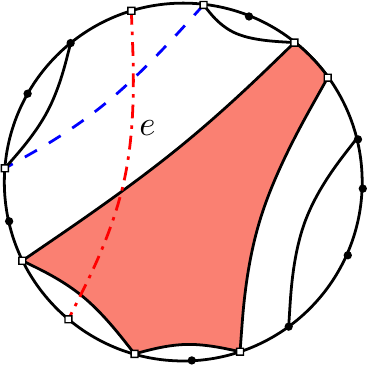}
    \label{fi:trico-algorithm-a}
	} 
    \hspace{0.1\textwidth}
    \subfigure[]{
    \includegraphics[height=0.25\textwidth] {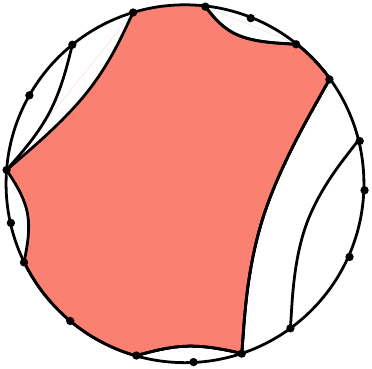}
    \label{fi:trico-algorithm-b}
	} 
\caption{
Two consecutive steps of Algorithm~\textsc{Straight-Line-Triconnected-Decision}.
(a) The outerplane graph $G_f$; the shaded face is \texttt{full} (the others 
are \texttt{empty}); the dash-dot edge $e$ is the next edge of $E_f$ to be 
considered; edges in $E_\chi$ are drawn as dashed lines; white squares are 
vertices of $V_\chi$. (b) Graph $G_f$ after the update due to edge $e$.
}
\label{fi:trico-algorithm}
\end{figure}

\medskip\noindent{\bf Algorithm}~\textsc{Straight-Line-Triconnected-Decision}. 
Let $\cal E$ be the unique planar embedding of $S$ (up to a flip). The 
algorithm verifies that each edge of $E \setminus W$ satisfies Condition~1 of 
Theorem~\ref{th:triconnected-decision} and that there exists a face $f$ of $\cal 
E$ containing three vertices $v_1$, $v_2$, and $v_3$, that satisfy Condition~2 of 
Theorem~\ref{th:triconnected-decision}. If both conditions hold, then $v_1$, 
$v_2$, and $v_3$ can be used to find a straight-line \cdrawing $\Gamma$ of 
$\langle G,S \rangle$ as described in the proof of 
Theorem~\ref{th:triconnected-decision}.

Condition~1 is verified as follows. Construct an auxiliary graph $S'$ from $S$ by subdividing each edge $e$ of $W$ with a dummy vertex $v_e$. Also, for each face $f$ of $\cal E$ add to $S'$ a vertex $v_f$ and connect $v_f$ to all non-dummy vertices of $f$. We have that two vertices of $V$ belong to the same face of $\cal E$ if and only if their distance in $S'$ is two. 

To test Condition~2 of Theorem~\ref{th:triconnected-decision} we perform the 
following procedure on each face $f$ of $\cal E$, restricting our attention to 
the set $E_f$ of edges in $E \setminus W$ whose end-vertices belong to $f$.
We maintain an auxiliary outerplane graph $G_f$ whose vertices are the vertices 
$V_f$ of $f$. Each internal face of $G_f$ is either marked as \texttt{full} or 
as \texttt{empty}. Faces marked \texttt{full} are not adjacent to each other. 
Intuitively, we have that any three vertices of an \texttt{empty} face satisfy 
Condition~2, while no triple of vertices of a \texttt{full} face satisfies such a condition.
We initialize $G_f$ with the cycle composed of the vertices and the edges of  $f$ and mark its unique internal face as \texttt{empty}. At each step an edge  $e$ of $E_f$ is considered and $G_f$ is updated accordingly. 
If the end-vertices of $e$ belong to a single \texttt{empty} face of $G_f$, then we update $G_f$ by splitting such a face into two \texttt{empty} faces. If the end-vertices of $e$ belong to a single \texttt{full} face, then we ignore $e$, as adding $e$ to $G_f$ would determine crossings between $e$ and several edges and faces (see Fig.~\ref{fi:trico-algorithm-a}). Consider the set $E_\chi$ of internal edges of $G_f$ crossed by $e$. Define a set of vertices $V_\chi$ of $G_f$ containing the end-vertices of $e$, the end-vertices of edges of $E_\chi$ that are incident to two \texttt{empty} faces, and the vertices of the \texttt{full} faces traversed by $e$. Remove all edges in $E_\chi$ from $G_f$. Mark the face $f'$ obtained by such a removal as \texttt{empty}. Form a new face $f_\chi$ inside $f'$ with all vertices in $V_\chi$ by connecting them as they appear in the circular order of $f$, and mark $f_\chi$ as \texttt{full} (see Fig.~\ref{fi:trico-algorithm-b}).

When all the edges of $E_f$ have been considered, if $G_f$ has an internal face 
marked as \texttt{empty}, any three vertices of this face satisfy Condition~2. 
Otherwise, $G_f$ has a single internal face marked \texttt{full} and no triple of 
vertices of $f$ satisfies Condition~2.

\begin{theorem}\label{th:triconnected-algorithm}
Let $G(V,E)$ be a graph and let $S(V,W)$ be a planar triconnected spanning subgraph of $G$. There exists an $O(|V| \times |E \setminus W|)$-time algorithm that decides whether $\langle G,S \rangle$ admits a straight-line \cdrawing $\Gamma$ and, in the positive case, computes it on an $O(|V|^2) \times O(|V|^2)$ grid. 
\end{theorem} 

\begin{proof}
First, apply Algorithm~\textsc{Straight-Line-Triconnected-Decision} to decide whether $\langle G,S \rangle$ satisfies the two conditions of Theorem~\ref{th:triconnected-decision}. If this is the case, apply the algorithm described in the proof of Theorem~\ref{th:triconnected-decision} to construct a straight-line \cdrawing $\Gamma$ of $\langle G,S \rangle$.

\medskip
\noindent\textsc{Correctness.}
Suppose that the \textsc{Straight-Line-Triconnected-Decision} algorithm verifies Condition~1 of Theorem~\ref{th:triconnected-decision} and identifies a face $f$ of $\cal E$ such that a face~$f_e$ marked \texttt{empty} can be found in the auxiliary outerplane graph $G_f$. Since $f_e$ is not traversed by any edge of $E_f$, we have that any three vertices of $f_e$ verify Condition~2 of Theorem~\ref{th:triconnected-decision} and any convex drawing of $S$ where such three vertices are incident to the external face corresponds to a \cdrawing of $\langle G,S \rangle$.  

Conversely, suppose that the \textsc{Straight-Line-Triconnected-Decision} algorithm does not verify Condition~1 or that the computation for every face $f$ of $\cal E$ yields a single internal face of the auxiliary outerplane graph $G_f$ marked \texttt{full}. Then we prove that Condition~2 of Theorem~\ref{th:triconnected-decision} is not verified. Our proof is based on the invariant that during the computation of Algorithm~\textsc{Straight-Line-Triconnected-Decision}, as well as at its end, any face of $G_f$ marked \texttt{full} cannot contain three vertices that satisfy Condition~2 of Theorem~\ref{th:triconnected-decision}. In order to see this, first observe that a pair of vertices that separates the end-vertices of some edge in the circular order of a face $f$ of $\cal E$ is also separated by such an edge. Hence, searching for three vertices of $f$ that do not separate the end-vertices of any edge of $E_f$ is equivalent to searching for three vertices that are not separated by any edge in $E_f$.  We prove the following statement: any pair of vertices of a \texttt{full} face $f_{full}$ that is non-contiguous in the circular order of $f_{full}$ is separated by an edge of $E_f$. Since a \texttt{full} face has at least four vertices, this implies that Condition~2 of Theorem~\ref{th:triconnected-decision} cannot be satisfied by any three vertices of $f_{full}$.

We first need to prove the following: any edge that crosses the boundary of a \texttt{full} face $f_{full}$ also crosses one edge of $E_f$. We prove this inductively on the composition of $f_{full}$. In the base case, when $f_{full}$ is created, the current edge $e$ of $E_f$ crosses a set of edges $E_\chi$ separating \texttt{empty} faces (i.e., no previous \texttt{full} face involved). In this case, any vertex of $V_\chi$ is the end-vertex of an edge in $E_\chi$ or of $e$, implying the statement. In the general case, when $f_{full}$ is created, the current edge $e$ of $E_f$ crosses both a set of edges $E_\chi$ separating \texttt{empty} faces and a set of \texttt{full} faces of $G_f$. Applying an inductive argument it can be proved that all vertices of $V_\chi$ are the end-vertices of some edge of $E_f$, thus proving the statement.

Finally, we prove inductively that any pair of vertices of a \texttt{full} face $f_{full}$ that is non-contiguous in the circular order of $f_{full}$ is separated by an edge of $E_f$. As above, in the base case (when face $f_{full}$ is created) the current edge $e=(v_1,v_2)$ of $E_f$ does not cross any \texttt{full} face. Since edges in $E_\chi$ do not cross each other, it is easy to verify that any pair of vertices of $V_\chi$ that are non-contiguous in the circular order of $f_{full}$ are separated either by an edge in $E_\chi$ or by $e$. 
Second, suppose that when face $f_{full}$ is created the current edge $e$ of $E_f$ crosses both a set of edges $E_\chi$ separating \texttt{empty} faces and a set of \texttt{full} faces of $G_f$. Consider a pair of vertices $u,v$ of $V_\chi$ that are non-contiguous in the circular order of $f_{full}$. Observe that, if $u$ and $v$ are separated by some edge, adding edge $(u,v)$ would introduce a crossing. If both $u$ and $v$ belong to the same \texttt{full} face that is merged into $f_{full}$, then by applying an inductive argument we have that they are separated by an edge of $E_f$. Otherwise, either edge $(u,v)$ would cross the boundary of at least one \texttt{full} face or it would cross $e$. In both cases edge $(u,v)$ would cross an edge of $E_f$, proving that $u$ and $v$ are separated by some edge of $E_f$.    

\medskip
\noindent\textsc{Time and Area Requirements.} 
The unique planar embedding $\cal E$ of $S$ (up to a flip) can be computed in 
$O(|V|)$ time. 
Regarding the time complexity of testing Condition~1, the auxiliary graph $S'$ can be constructed in 
time linear in the size of $S$. Since $S'$ is a planar graph, deciding if two 
vertices have distance two can be done in constant time provided that an 
$O(|V|)$-time preprocessing is performed~\cite{kk-spqpgct-03}. Thus, testing 
Condition~1 for all edges in $E \setminus W$ can be done in $O(|V| + |E 
\setminus W|)$ time. 
Regarding the time complexity of testing Condition~2, observe that, for each 
face $f$ of $\cal E$, $E_f$ can be computed in $O(|V| + |E \setminus W|)$ time 
while verifying Condition~1.
Since for each face $f$ of $\cal E$, the size of $G_f$ is $O(|V_f|)$, adding 
edges in $E_f$ has time complexity $O(|E_f| \times |V_f|)$. Overall, as $\sum_{f 
\in {\cal E}}{|E_f|} = O(|E \setminus W|)$ and $\sum_{f \in {\cal E}}{|V_f|} = 
O(|V|)$, the time complexity of testing Condition~2 is $O(|E 
\setminus W| \times |V|)$, which gives the time complexity of the whole 
algorithm.
Regarding the area, the algorithm in~\cite{br-scdpg-06} can be used to obtain in 
linear time a straight-line grid drawing of $S^*$ on an $O(|V|^2) \times 
O(|V|^2)$ grid.
\end{proof}


\section{Polyline Drawings}\label{se:polyline}

In this section we allow bends along the edges of $G$ not in $S$, while we still require that the edges of $S$ are drawn as straight-line segments. Of course, since edge bends are negatively correlated to the drawing readability, our goal is to compute $k$-bend \cdrawings for small values of $k$. However, it might happen that the number of bends in the drawing can only be reduced at the cost of increasing the required area. Throughout the section, we discuss possible trade-offs between these two measures of the quality of the drawing. 

We split the section into two subsections, dealing with the case in which $S$ is a spanning tree and with the case in which $S$ is a biconnected spanning graph, respectively.

\subsection{Spanning Trees}\label{se:bend-tree}

In this subsection we prove that allowing bends along the edges of $G$ not in $S$ permits us to compute  \cdrawings of pairs $\langle G,S \rangle$ for every spanning tree $S$ of $G$; such drawings are realized on a polynomial-area integer grid. We provide algorithms that offer different trade-offs between number of bends and drawing area.

Let $G(V,E)$ be a graph with $n$ vertices and $m$ edges, and let $S(V,W)$ be any spanning tree of $G$. 
We denote by $x(v)$ and $y(v)$ the $x$- and the $y$-coordinate of a vertex $v$, respectively.
The following algorithm computes a $1$-bend \cdrawing of $\langle G,S \rangle$.

\medskip\noindent{\bf Algorithm}~\textsc{One-bend Tree}. 
The algorithm works in two steps (refer to Fig.~\ref{fi:polyline-trees-1-bend}).

\smallskip\noindent{\textsc{Step~1:}} 
Consider a point set of size $n$ such that for each point $p_i$, the $x$- and
$y$-coordinates of $p_i$ are $i^2$ and $i$, respectively. 
Construct a straight-line drawing of $S$ by placing the
vertices on points $p_i$, $1 \leq i \leq n$, according to a DFS traversal.

\smallskip\noindent{\textsc{Step~2:}} 
Let $v_i$ be the vertex placed on point $p_i$. For each $i \in \{1,\ldots,n\}$, draw each edge
$(v_i,v_j) \in E \setminus W$ such that $j>i$ as a polyline connecting $p_i$
and $p_j$, and bending at point $(i^2+1, n+c)$ where $c$ is a progressive
counter, initially set to one.

\begin{figure}[tb!]
  \centering
    \subfigure[]{
    	\includegraphics[height=33mm] {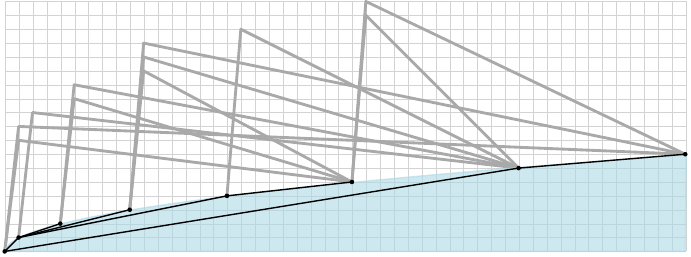}
    	\label{fi:polyline-trees-1-bend}
    } 
    \subfigure[]{
	    \includegraphics[height=33mm] {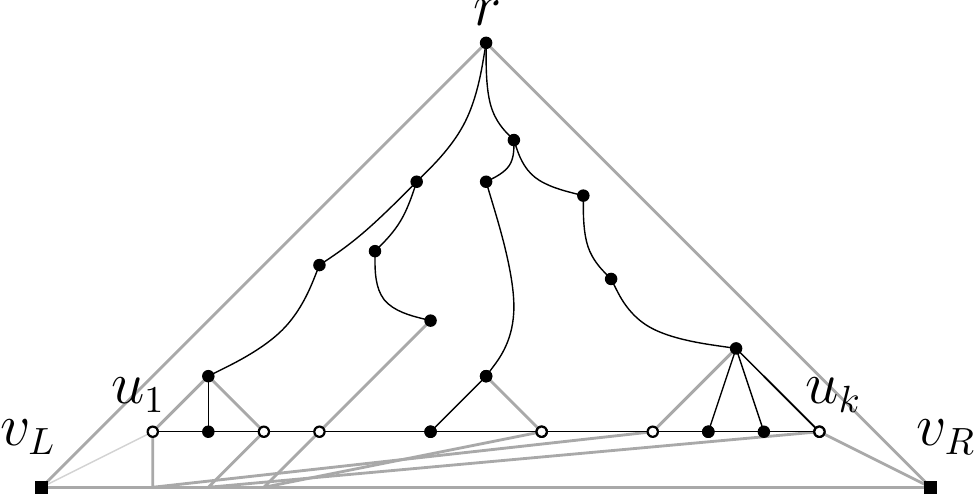}
	    \label{fi:polyline-trees-3-bend}
     } 
\caption{Illustration of: (a) Algorithm~\textsc{One-bend Tree} and (b) Algorithm~\textsc{Three-bend Tree}; a graph $G$ with a given spanning tree $S$ (black edges).
}
\end{figure}

\begin{theorem}\label{th:trees-1-bend}
Let $G(V,E)$ be a graph with $n$ vertices and $m$ edges, and let $S(V,W)$ be
any spanning tree of $G$. There exists a 1-bend \cdrawing $\Gamma$ of $\langle
G,S \rangle$. Drawing $\Gamma$ can be computed in $O(n+m)$ time and has
$O(n^2(n+m))$ area. 
\end{theorem}

\begin{proof}
The algorithm that constructs $\Gamma$ is Algorithm~\textsc{One-bend Tree}. In 
the following we first prove that $\Gamma$ is a straight-line \cdrawing of 
$\langle G,S \rangle$, and then we analyze the time and area complexity. We 
adopt the same notation used in the description of the algorithm.\medskip

\noindent\textsc{Correctness.} 
We have to prove that in $\Gamma$ $(i)$: the edges of $S$ are never crossed and 
$(ii)$ there exists no overlapping between a bend-point and an edge in $E 
\setminus W$.
To prove $(i)$, observe that the drawing of $S$ contained in $\Gamma$ is planar 
and that the edges in $E \setminus W$ are drawn outside the convex region 
containing the drawing of $S$. 
To prove $(ii)$, observe that, for each two edges $(v_i,v_j)$ and $(v_p,v_q)$ 
such that $i<j$, $p<q$, and $p<i$, the bend-point of the polyline representing 
$(v_i,v_j)$ lies above the polyline representing $(v_p,v_q)$. 

\medskip

\noindent\textsc{Time and area requirement.} 
Concerning time complexity, \textsc{Step~1} can be performed in $O(n)$ time. 
\textsc{Step~2} can be performed in $O(m)$ time, since for each edge in $E 
\setminus W$ a constant number of operations is required.
Concerning area requirements, the width of $\Gamma$ is $O(n^2)$, by 
construction, while the height of $\Gamma$ is given by the $y$-coordinate of the 
topmost bend-point, that is $n+m$.

\end{proof}

Next, we describe an algorithm that constructs $3$-bend \cdrawings
of pairs $\langle G, S \rangle$ with better area bounds than
Algorithm~\textsc{One-bend Tree} for sparse graphs.

\medskip\noindent{\bf Algorithm}~\textsc{Three-bend Tree}. 
The algorithm works in four steps (refer to Fig.~\ref{fi:polyline-trees-3-bend} for an illustration).

\smallskip\noindent{\textsc{Step~1:}} 
Let $G'$ be the graph obtained from $G$ by subdividing each edge $(v_i,v_j) \in E \setminus W$ with two dummy vertices $d_{i,j}$ and $d_{j,i}$.
Let $S'$ be the spanning tree of $G'$, rooted at any non-dummy vertex $r$,
obtained by deleting all edges connecting two dummy vertices.
Clearly, every dummy vertex is a leaf of $S'$. 

\smallskip\noindent{\textsc{Step~2:}} 
For each vertex of $S'$, order its children arbitrarily, thus
inducing a left-to-right order of the leaves of $S'$. Rename the leaves of $S'$
as $u_1,\dots,u_k$ following this order. For each $i \in \{1,\dots,k-1\}$, add an
edge $(u_i,u_{i+1})$ to $S'$. Also, add to $S'$ two dummy vertices $v_L$ and
$v_R$, and edges $(v_L,r)$,$(v_R,r)$,$(v_L,u_1)$,$(u_k,v_R)$, $(v_L,v_R)$.

\smallskip\noindent{\textsc{Step~3:}} 
Construct a straight-line grid drawing $\Gamma'$ of $S'$, as described in~\cite{kant-dpguco-96}, in which edge $(v_L,v_R)$ is drawn as a horizontal segment on the outer face, vertices $u_1,\dots,u_k$ all lie on points having the same $y$-coordinate $Y$, and the rest of $S'$ is drawn above such points.
Remove from $\Gamma'$ the vertices and edges added in \textsc{Step~2}.

\smallskip\noindent{\textsc{Step~4:}}
Compute a drawing $\Gamma$ of $G$ such that each edge in $W$ is drawn as in
$\Gamma'$, while each edge $(v_i,v_j) \in E \setminus W$ is drawn as a polyline
connecting $v_i$ and $v_j$, bending at $d_{i,j}$, at $d_{j,i}$, and at a
point $(c,Y-1)$ where $c$ is a progressive counter, initially set to $x(u_1)$. 

%
\begin{theorem}\label{th:trees-3-bend}
Let $G(V,E)$ be a graph with $n$ vertices and $m$ edges, and let $S(V,W)$ be any spanning tree of $G$. 
There exists a 3-bend \cdrawing $\Gamma$ of $\langle G,S \rangle$. 
Drawing $\Gamma$ can be computed in $O(n+m)$ time and has $O((n+m)^2)$ area. 
\end{theorem}

\begin{proof}
The algorithm that constructs $\Gamma$ is Algorithm~\textsc{Three-bend Tree}.
In the following we first prove that $\Gamma$ is a straight-line \cdrawing of
$\langle G,S \rangle$, and then we analyze the time and area complexity. We
adopt the same notation used in the description of the algorithm.\medskip

\noindent\textsc{Correctness.} 
We have to prove that in $\Gamma$ $(i)$ the edges of $S$ are never crossed and 
$(ii)$ there exists no overlapping between a bend-point and an edge in $E 
\setminus W$.
To prove $(i)$, observe that the drawing of $S$ contained in $\Gamma$ is 
planar~\cite{kant-dpguco-96} and lies above the horizontal line with 
$y$-coordinate $Y$. Also, observe that each edge-segment that is drawn by 
{\textsc{Step~3}} and {\textsc{Step~4}} either lies below such line or has the 
same representation as an edge of $S'$ in $\Gamma'$.
To prove $(ii)$, observe that, for each edge in $E \setminus W$, the first and 
the last bend-points have the same position as dummy vertices of $S'$ in 
$\Gamma'$, while the second bend-point has the lowest $y$-coordinate in $\Gamma$ 
and no two points with the same $y$-coordinate are connected by a straight-line 
segment.

\medskip

\noindent\textsc{Time and area requirement.} 
Concerning time complexity, \textsc{Step~1} can clearly be performed in $O(m)$ 
time; \textsc{Step~2} in $O(n+m)$ time; \textsc{Step~3} in $O(n+m)$ 
time~\cite{kant-dpguco-96}; and \textsc{Step~4} in $O(n+m)$ time, since for each 
edge in $E$ a constant number of operations is required.
Concerning area requirements, by construction and by the requirements of the
algorithm in~\cite{kant-dpguco-96}, both the width and the height of $\Gamma$
are $O(n+m)$.

\end{proof}

We finally prove that there exists a drawing algorithm that computes $4$-bend  
\cdrawings that are more readable than those computed by 
Algorithm~\textsc{Three-bend Tree}. Although the area of these drawings is still 
$O((n+m)^2)$,  they have optimal crossing angular resolution, i.e., edges cross 
only at right angles. 
Drawings of this type, called \emph{RAC drawings}, have been widely studied in the literature~\cite{DBLP:journals/tcs/DidimoEL11,dl-cargd-12} and are motivated by cognitive studies showing that drawings where the edges cross at very large angles do not affect too much the drawing  readability~\cite{DBLP:conf/apvis/HuangHE08}.

Let $G(V,E)$ be a graph and let $S(V,W)$ be any spanning tree of $G$ (see, 
e.g., 
Fig.~\ref{fi:polyline-trees-4-bends-a}). The drawing algorithm that computes a 
$4$-bend \cdrawings of $\langle G,S \rangle$ is described below.

\medskip\noindent{\bf Algorithm}~\textsc{Four-bend Tree}. The algorithm   works 
in three steps:

\smallskip\noindent{\textsc{Step~1:}} Root $S$ at any  vertex and subdivide 
each edge $e=(u,v) \in E \setminus W$ with two dummy vertices $u_e$ and $v_e$. 
Let $G'(V',E')$ be the obtained graph and let $S'(V',W')$ be the spanning tree 
of $G'$ obtained by deleting all edges connecting two dummy vertices. Clearly, 
every dummy vertex is a leaf of $S'$. For each vertex $u \in V'$, order the 
children of $u$ arbitrarily, but such that all its dummy children appear before 
the real ones from left to right (see Fig.~\ref{fi:polyline-trees-4-bends-b}).

\smallskip\noindent{\textsc{Step~2:}} Visit $S'$ in post  order and for each 
vertex $u$ of $S'$, define a weight $\omega(u)$ as follows: if $u$ is a leaf, 
then $\omega(u)=0$; if $u$ is an internal vertex with children $v_1, v_2, 
\dots, 
v_k$, then $\omega(u) = \sum_{i=1}^k (\omega(v_i) + 1) - 1$. Draw $S'$ 
downward, 
so that vertices of the same level in the tree are at the same $y$-coordinate, 
and such that $|y(u) - y(v)| = 1$ if $u$ and $v$ belong to two consecutive 
levels. Denote by $y_{\min}$ the minimum $y$-coordinate of a vertex. Visiting 
$S'$ top-down and left-to-right, assign the $x$-coordinate to each vertex as 
follows: if $u$ is the root, then $x(u)=0$; if $v_1, v_2, \dots, v_k$ are the 
children of an internal vertex $u$, then $x(v_1)=x(u)$ and $x(v_i) = x(v_{i-1}) 
+ \omega(v_{i-1}) + 1$ $(2 \leq i \leq k)$. The edges of $S'$ are drawn as 
straight-line segments, while each edge $e=(u_e,v_e) \in E'\setminus W'$ is 
drawn as a polyline $u_e, a, b, v_e$, where $a$ and $b$ are two points such 
that 
$x(a)=x(u_e)$, $x(b)=x(v_e)$ and $y(a)=y(b)=y_{\min} - c$, where $c$ is a 
progressive counter, initially set to one. Refer to 
Fig.~\ref{fi:polyline-trees-4-bends-c} for an illustration.

\smallskip\noindent{\textsc{Step~3:}} Replace each dummy  vertex with a 
bend-point to get a 4-bend \cdrawing of $\langle G,S \rangle$. See 
Fig.~\ref{fi:polyline-trees-4-bends-d} for an illustration of the resulting 
drawing.
 
\begin{figure}[tb!]
  \centering
  	\subfigure[\textsc{Input}]{
    \includegraphics[width=0.44\textwidth] {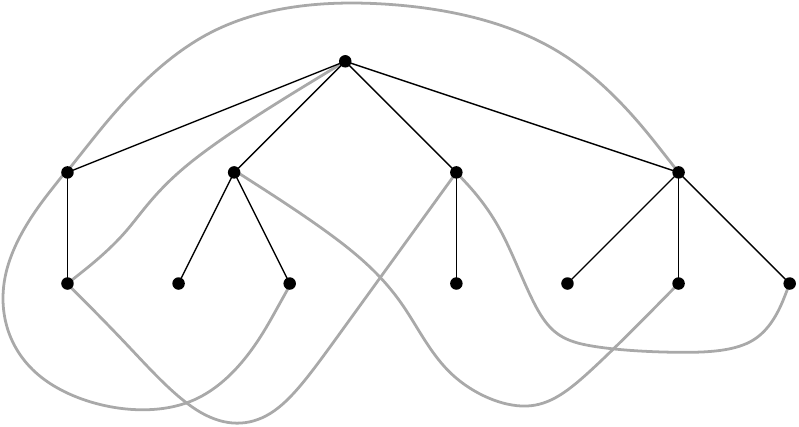}
    \label{fi:polyline-trees-4-bends-a}
	} 
    \hspace{0.02\textwidth}
    \subfigure[\textsc{Step~1}]{
    \includegraphics[width=0.44\textwidth] {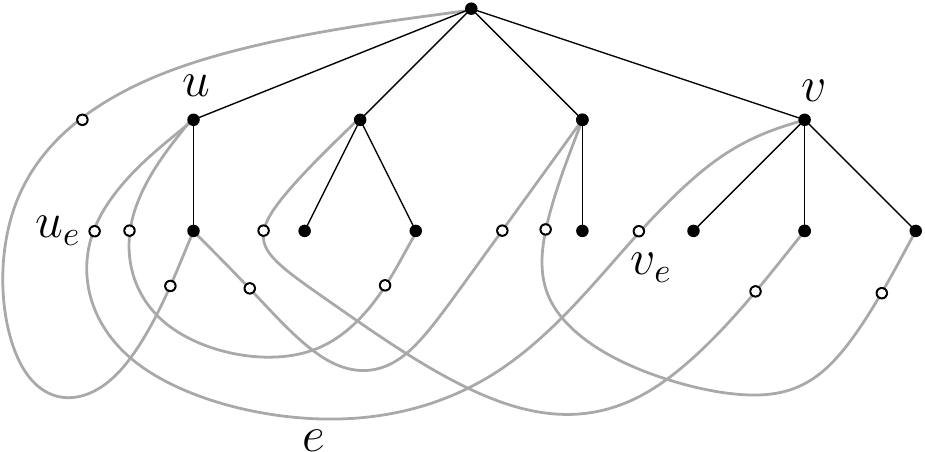}
    \label{fi:polyline-trees-4-bends-b}
	} 
	\subfigure[\textsc{Step~2}]{
    \includegraphics[width=0.44\textwidth] {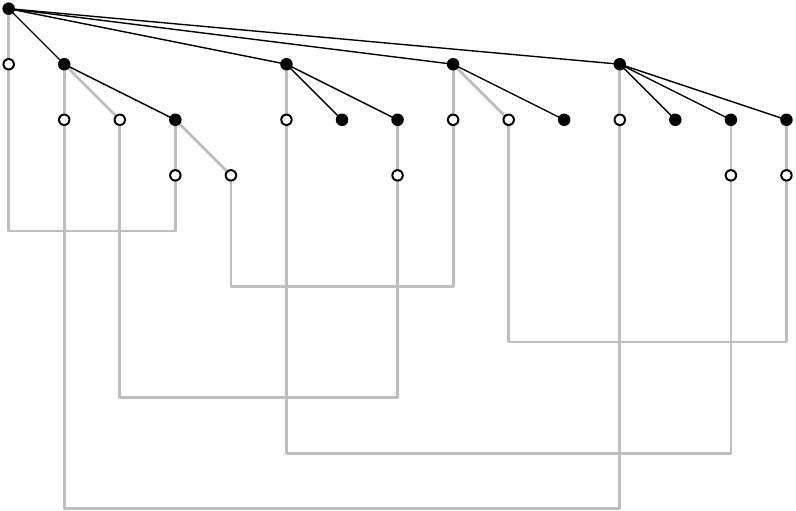}
    \label{fi:polyline-trees-4-bends-c}
	} 
	\hspace{0.02\textwidth}
	\subfigure[\textsc{Step~3}]{
    \includegraphics[width=0.44\textwidth] {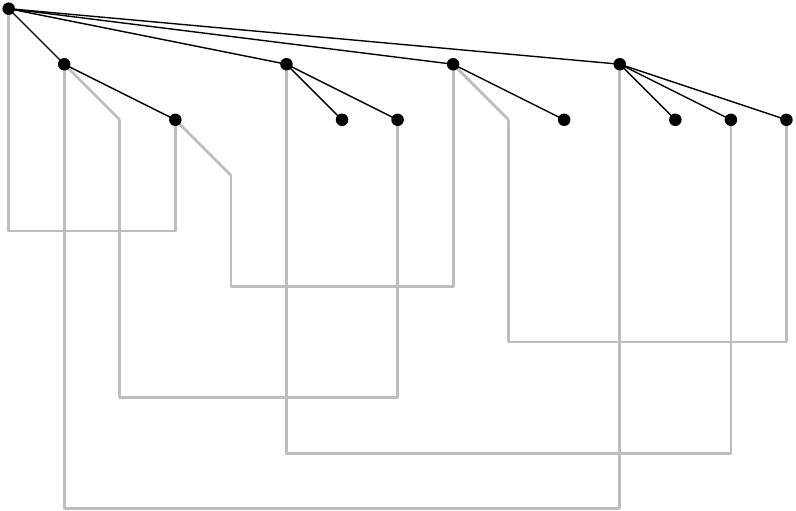}
    \label{fi:polyline-trees-4-bends-d}
	} 
\caption{Illustration of Algorithm~\textsc{Four-bend-Tree}: (a) a graph  $G$ 
with a given spanning tree $S$ (black edges); (b) an enhanced version $G'$ of 
$G$, where each edge that does not belong to $S$ is subdivided with two dummy 
vertices (white circles); (c) a 4-bend \cdrawing of $\langle G',S' \rangle$; (d) 
a 4-bend \cdrawing of $\langle G,S \rangle$. 
}
\label{fi:polyline-trees-4-bends}
\end{figure}

\begin{theorem}\label{th:trees-4-bends}
Let $G$ be a graph with $n$ vertices and $m$ edges, and let $S$ be any  
spanning tree of $G$. There exists a grid 4-bend \cdrawing $\Gamma$ of $G$, 
which is also a RAC drawing. Drawing $\Gamma$ can be computed in $O(n+m)$ time 
and has $O((n+m)^2)$ area.
\end{theorem}
\begin{proof}
The algorithm that constructs $\Gamma$ is Algorithm~\textsc{Four-bend-Tree}.  
In 
the following we first prove that $\Gamma$ is a straight-line compatible 
drawing 
of $\langle G,S \rangle$, and then we analyze time complexity and area 
requirement. We adopt the same notation used in the description of the 
algorithm.\medskip

\noindent\textsc{Correctness.} The use of dummy vertices and the way in which  
the vertex coordinates are defined guarantee that each edge $e \in E' 
\setminus 
W'$ is incident to any of its end-vertices with a vertical segment that does 
not 
intersect any other element in the sub-drawing of $S'$, and hence of $S$. Also, 
since every bend of $\Gamma$ is drawn below the drawing of $S'$ (and hence of 
$S$), the horizontal segments that connect two bends do not cross the edges of 
$S$. Finally, it is immediate to see that each edge of $W$ has at most four 
bends, and that its vertical segments can only intersect horizontal segments, 
thus $\Gamma$ is a RAC \cdrawing of $\langle G,S \rangle$.

\noindent\textsc{Time and area Requirement.} By construction, $\Gamma$ requires
at most one column for each vertex and two columns for each edge. Also, it 
requires at most $n$ rows for the sub-drawing of $S$ (namely, one row for each 
level), and one row for each edge in $E \setminus W$. Thus, $\Gamma$ takes 
$O((n+m)^2)$ area. 

About the time complexity, \textsc{Step~1} can be easily executed in linear  
time and adds $O(m)$ dummy vertices. In \textsc{Step~2}, all the vertices are 
drawn in $O(n')=O(n+m)$ time, where $n'$ denotes the number of vertices of 
$G'$. 
Also, all the edges are drawn in $O(m')=O(m)$, where $m'$ is the number of 
edges 
of $G'$. Finally, the removal of dummy vertices in \textsc{Step~3} takes $O(m)$ 
time. Hence, the time-complexity bound follows.   
\end{proof}

\subsection{Biconnected Spanning Subgraphs}\label{se:biconnected}

In this subsection we consider the case in which $S(V,W)$ is a biconnected spanning subgraph. We note that, in this case, an example of a negative instance $\langle G, S \rangle$ can be directly inherited from the case in which $S$ is a triconnected graph (Section~\ref{sse:straight-line-triconnected}), as any triconnected graph is also a biconnected graph. More generally, a trivial necessary condition for an instance to be positive is that $S$ admits an embedding in which for every edge $e$ of $E \setminus W$ the two end-vertices of $e$ share a face. In the following we prove that this condition is also sufficient if the edges of $E \setminus W$ are allowed to bend, hence providing a characterization of the positive instances. Then, in the rest of the section, we propose algorithms to realize \cdrawings of such instances, discussing the trade-off between the number of bends and the area of the resulting drawings.

\begin{theorem}\label{th:k-bend-biconnected-decision}
Let $G(V,E)$ be a graph and let $S(V,W)$ be a spanning biconnected planar 
subgraph of $G$. There exists a positive integer $k=O(n)$ such that 
$\langle G, S\rangle$ admits a $k$-bend \cdrawing $\Gamma$ if and only if 
there exists a planar (combinatorial) embedding $\cal E$  of $S$ such that each edge $e \in E\setminus W$ connects two vertices belonging 
to the same face of~$\cal E$.  
\end{theorem}

\begin{proof}
We prove the necessity. Suppose that $\langle G,S \rangle$ admits a 
$k$-bend \cdrawing $\Gamma$ for some positive integer $k=O(n)$. 
Consider the drawing $\Gamma'$ of $S$ contained in
$\Gamma$ and let $\cal E$ be the planar embedding of $S$ determined by 
$\Gamma'$. For each edge $e\in E\setminus W$, the polyline representing $e$ 
in $\Gamma$ must be drawn inside a face of $\cal E$, as otherwise a crossing 
between $e$ and some edge of $S$ would occur.

We prove the sufficiency. Suppose that $S$ admits a planar embedding $\cal E$ 
such that each edge $e \in E\setminus W$ is incident to vertices belonging 
to the same face of~$\cal E$. We construct a $k$-bend \cdrawing $\Gamma$ 
of $\langle G,S \rangle$ as follows. Produce a straight-line planar drawing $\Gamma^*$ of $S$ whose embedding is $\cal E$. For each edge $e=(u,v)$ of $E\setminus W$, let $f_{e}$ be any face of $\Gamma^*$ containing both the end-vertices of $e$, and let $p$ be one of the two paths of $f_{e}$ between $u$ and $v$.   
Draw $e$ inside the polygon representing 
$f_{e}$ in $\Gamma^*$ as a polyline starting at $u$, ending at $v$, bending 
at points arbitrarily close to each vertex of $p$, and such that
there exists no overlap between a bend-point and an edge. Observe that, 
each polyline representing an edge of $E\setminus W$ has at most $\lceil 
\frac{n}{2} \rceil-2$ bend-points. This concludes the proof of the theorem.
\end{proof}

Theorem~\ref{th:k-bend-biconnected-decision} presented a simple characterization 
of those pairs $\langle G,S\rangle$ that admit a $k$-bend \cdrawing provided 
that $k$ can be arbitrarily large, although linear in the size of $S$.
In the following we describe an algorithm that tests in polynomial time whether 
the condition of Theorem~\ref{th:k-bend-biconnected-decision} is satisfied by a 
pair $\langle G,S \rangle$ in which $S$ is biconnected and spanning and, if this is the case, returns an embedding $\cal E$ of $S$ satisfying such condition.

Our algorithm exploits a reduction to instances of problem Simultaneous Embedding with Fixed Edges with Sunflower Intersection ({\sc Sunflower SEFE})~\cite[Chapter 11]{bkr-sepg-12}. Problem {\sc Sunflower SEFE} takes as input a set of $k$ planar graphs $G_1, \dots, G_k$ on the same set of vertices such that, for each two graphs $G_i$ and $G_j$, with $i \neq j$, $G_i \cap G_j = G_\cap$, where $G_\cap = \bigcap_{l=1}^k G_l$. Namely, the intersection graph is the same for each pair of input graphs. Problem {\sc Sunflower SEFE} asks whether $G_1, \dots, G_k$ admit planar drawings $\Gamma_1, \dots, \Gamma_k$, respectively, on the same point set such that each edge $e \in G_\cap$ is represented by the same curve in each of $\Gamma_1, \dots, \Gamma_k$.

Problem {\sc Sunflower SEFE} has been shown NP-complete for $k=3$ and $G_\cap$ being a spanning tree~\cite{adn-ocsprs-13}. However, in our reduction, the produced instances of the problem will always have the property that $G_\cap$ is biconnected and each graph $G_i$ ($1 \leq i \leq k$) is composed only of the edges of $G_\cap$ plus a single edge. In this setting, a polynomial-time algorithm can be easily derived from the known algorithms~\cite{adfpr-tsetgibgt-11,bkr-seeorpc-13,hrl-tspcg2c-13} that solve in polynomial time the case in which $k=2$ and $G_\cap$ is biconnected.

\medskip\noindent{\bf Algorithm}~\textsc{K-Bend-Biconnected-Decision}.
The algorithm tests whether $S$ admits an embedding $\cal E$ such that the 
end-vertices of each edge $e \in E\setminus W$ belong to the same face of~$\cal E$ 
in two steps, as follows.

First, instance  $\langle G(V,E),S(V,W) \rangle$ is reduced to an instance $\langle 
G_1,\dots,G_{| E \setminus W| }\rangle$ of {\sc Sunflower SEFE}. For each edge 
$e_i$ with $i=1,\dots,| E \setminus W|$, the reduction simply sets $G_i = S 
\cup e_i$. By construction, the intersection graph $G_\cap$ of the constructed instance coincides with $S$ (and hence is biconnected) and each graph $G_i$ is composed only of the edges of $G_\cap$ plus a single edge.

Second, the existence of drawings $\Gamma_1,\dots,\Gamma_{| E \setminus W|}$ of $\langle G_1,\dots,G_{| E \setminus W|}\rangle$ is tested by means of the algorithm by Angelini \emph{et al.}~\cite{adfpr-tsetgibgt-11}, which also constructs such drawings if a {\sc Sunflower SEFE} exists. The embedding $\cal E$ of $S$ is obtained by restricting any of the drawings $\Gamma_i$ to the edges of $G_\cap$ (and hence of $S$). Note that, in order for $\Gamma_1,\dots,\Gamma_{| E \setminus W|}$ to be a solution of the {\sc Sunflower SEFE} instance, the embedding of $G_\cap$ resulting by restricting $\Gamma_i$ to its edges is the same for every $i = 1, \dots, |E \setminus W|$.

\begin{lemma}\label{le:k-bend-biconnected-algorithm}
Let $G(V,E)$ be a graph and let $S(V,W)$ be a biconnected planar spanning subgraph of $G$. There exists an $O(|V|+|E\setminus W|)$-time algorithm that decides whether $S$ admits an embedding $\cal E$ such that the end-vertices of each edge $e \in E\setminus W$ belong to the same face of~$\cal E$ and, if this is the case, computes $\cal E$.
\end{lemma}

\begin{proof}
Apply Algorithm~\textsc{K-Bend-Biconnected-Decision}.

\medskip
\noindent\textsc{Correctness.} In order to prove the correctness of the algorithm, we show that an embedding $\cal E$ with the required properties exists if and only if the instance of {\sc Sunflower SEFE} constructed in the first step of Algorithm~\textsc{K-Bend-Biconnected-Decision} is positive.
Namely, it is easy to observe that in any {\sc Sunflower SEFE} of the $|E \setminus W|$  constructed graphs, the embedding of $S$ satisfies the condition of Theorem~\ref{th:k-bend-biconnected-decision}. For the other direction, it is sufficient to observe that, since each graph $G_i$ contains exactly one edge $e_i$ not belonging to $S$, once an embedding of $S$ with the property of Theorem~\ref{th:k-bend-biconnected-decision} has been computed, edge $e_i$ can be drawn inside a face of such embedding containing both its end-vertices without intersecting any edge of $G_i$, thus yielding a {\sc Sunflower SEFE} of $\langle G_1,\dots,G_{| E \setminus W| }\rangle$.

\noindent\textsc{Time Requirement.} The construction of the {\sc Sunflower SEFE} instance $\langle G_1,\dots,G_{| E \setminus W| }\rangle$ requires $O(|V|+|E\setminus W|)$ total time, as the description of $S$ does not need to be repeated in each data structure representing a graph $G_i$. Also, the second step of the algorithm can be performed in $O(|V|+|E\setminus W|)$ total time~\cite{adfpr-tsetgibgt-11}.
\end{proof}

The following corollary summarizes the results of Theorem~\ref{th:k-bend-biconnected-decision} and Lemma~\ref{le:k-bend-biconnected-algorithm}.

\begin{corollary}\label{co:k-bend-biconnected}
Let $G(V,E)$ be a graph and let $S(V,W)$ be a biconnected planar spanning subgraph of $G$. There exists an $O(|V|+|E\setminus W|)$-time algorithm that decides whether $\langle G, S \rangle$ admits a $k$-bends \cdrawing for some positive integer $k$.
\end{corollary}

In the following we further elaborate on the result given in Corollary~\ref{co:k-bend-biconnected}, by showing that, for the positive instances $\langle G, S \rangle$ there always exist 1-bend \cdrawings (with no guarantee on the area requirement) or 2-bend \cdrawings with polynomial area.

\medskip\noindent{\bf Algorithm}~\textsc{1-Bend-Biconnected-Drawing}.
Given an instance $\langle G,S \rangle$ satisfying the condition of Theorem~\ref{th:k-bend-biconnected-decision}, the algorithm constructs a 1-bend \cdrawing $\Gamma$ of $\langle G,S \rangle$ as follows.

First, let $\cal E$ be a planar embedding of $S$ such that each edge $e \in E\setminus W$ is incident to vertices belonging to the same face of~$\cal E$. Such embedding exists as $\langle G,S \rangle$ satisfies the condition of Theorem~\ref{th:k-bend-biconnected-decision}.

Then, for each face $f$ of $\cal E$, add to $S$ a dummy vertex $v_f$ inside $f$ and connect it to each vertex of $f$, hence obtaining a planar graph $S'$. Produce any straight-line planar drawing $\Gamma'$ of $S'$, and obtain a straight-line planar drawing $\Gamma^*$ of $S$ by removing from $\Gamma'$ all dummy vertices and their incident edges. Observe that, $\Gamma^*$ is a \emph{star-shaped} drawing, namely each face is represented by a polygon whose kernel is not empty (as it contains the point on which vertex $v_f$ has been placed and an infinite set of points around it). 

Finally, obtain $\Gamma$ from $\Gamma^*$ as follows. For each edge $e$ of $E\setminus W$, let $f_{e}$ be any face of $\Gamma^*$ containing both the end-vertices of $e$. Draw $e$ as a polyline whose unique bend-point is placed on a point inside the kernel of the polygon representing $f_{e}$ in $\Gamma^*$ in such a way that there exists no overlapping between a bend-point and an edge.

\begin{theorem}\label{th:1-bend-biconnected-algorithm}
Let $G(V,E)$ be a graph and let $S(V,W)$ be a biconnected planar spanning subgraph of $G$. There exists an $O(|V|+|E\setminus W|)$-time algorithm that decides whether $\langle G,S \rangle$ admits a 1-bend \cdrawing $\Gamma$ and, in the positive case, computes such a drawing.
\end{theorem}

\begin{proof}
First, apply the \textsc{K-Bend-Biconnected-Decision} algorithm to decide whether $\langle G,S \rangle$ satisfies the condition of Theorem~\ref{th:k-bend-biconnected-decision}. If this is the case, let $\cal E$ be the computed embedding of $S$. Then, apply Algorithm~\textsc{1-Bend-Biconnected-Drawing} to construct a 1-bend \cdrawing $\Gamma$ of $\langle G,S \rangle$ in which the embedding of $S$ is $\cal E$.

\medskip
\noindent\textsc{Correctness.} By Theorem~\ref{th:k-bend-biconnected-decision}, Algorithm~\textsc{K-Bend-Biconnected-Decision} correctly recognizes instances that do not admit any $k$-bend \cdrawing, and hence any 1-bend \cdrawing of $\langle G,S \rangle$. On the other hand, if Algorithm~\textsc{K-Bend-Biconnected-Decision} gives a positive answer, returning an embedding $\cal E$ with the required properties, Algorithm~\textsc{1-Bend-Biconnected-Drawing} correctly constructs a 1-bend \cdrawing $\Gamma$ of $\langle G,S \rangle$. Indeed, the addition of a dummy vertex inside every face of the computed embedding of $S$ can always be performed while maintaining the property that the resulting graph is planar and simple, which implies that a planar star-shaped drawing of $S$ preserving $\cal E$ can be constructed. By observing that the kernel of the polygon representing each face of $\cal E$ in such a drawing is a positive-area region containing an infinite number of points, it is easy to see that non-overlapping bend-points for all the edges traversing a face can be placed inside the kernel of this face.

\noindent\textsc{Time Requirement.} By Lemma~\ref{le:k-bend-biconnected-algorithm}, positive or negative instances can be recognized in $O(|V|+|E\setminus W|)$ time. If the instance is positive, a star-shaped drawing of $S$ can be found in linear time by applying any algorithm constructing planar drawings of planar graphs~\cite{s-epgg-90,dfpp-hdpgg-90}, as only a linear number of dummy vertices is added to $S$. Finally, each edge of $E\setminus W$ can be routed inside a face of $\cal E$ in constant time.
\end{proof}


Observe that the compatible drawings computed by Algorithm~\textsc{1-Bend-Biconnected-Drawing} are not guaranteed to have polynomial area, as the kernel of the polygon representing the face in which the bend-points are placed might be arbitrarily small (although containing an infinite number of points). In the following we prove that, when each edge of $E\setminus W$ is allowed to bend twice, it is possible to construct drawings requiring polynomial area in the size of the input.  

\medskip\noindent{\bf Algorithm}~\textsc{2-Bend-Biconnected-Drawing}.
Given an instance $\langle G,S \rangle$ satisfying the condition of Theorem~\ref{th:k-bend-biconnected-decision}, the algorithm constructs a 2-bend \cdrawing $\Gamma$ of $\langle G,S \rangle$ as follows.

First, let $\cal E$ be a planar embedding of $S$ such that each edge $e \in E\setminus W$ is incident to vertices belonging to the same face of~$\cal E$. Such embedding exists as $\langle G,S \rangle$ satisfies the condition of Theorem~\ref{th:k-bend-biconnected-decision}.

We construct a 2-bend \cdrawing $\Gamma$ of $\langle G,S \rangle$ as follows. For each face $f$ of $\cal E$, let $E_f$ be the set of edges of $E\setminus W$ that have to be drawn inside $f$ in $\Gamma$. Add to $S$ a cycle $c_f$ of length $2 |E_f|$ inside $f$ and connect its vertices to the vertices of $f$ in such a way that each vertex of $f$ is adjacent to consecutive vertices of $c_f$, and consecutive vertices of $c_f$ are adjacent to the same vertex of $f$ or to consecutive vertices of $f$, hence obtaining a planar graph $S'$. If $S'$ is not triconnected, augment it to planar triconnected by adding dummy edges.

Produce a straight-line grid convex drawing $\Gamma'$ of $S'$, as described in~\cite{br-scdpg-06}. Denote by $N_f(v)$ the set of grid points the neighbors of a vertex $v$ of $f$ belonging to a cycle $c_f$ have been placed on. Then, obtain a grid drawing $\Gamma^*$ of $S$ by removing from $\Gamma'$ all dummy vertices and all dummy edges.

Finally, obtain $\Gamma$ from $\Gamma^*$ as follows. For each face $f$ of $\Gamma^*$ and for each edge $e = (u,v) \in E_f$, draw $e$ as a polyline starting at $u$, bending at a point in $N_f(u)$ and at a point in $N_f(v)$, and ending at $v$.

\begin{theorem}\label{th:2-bend-biconnected-algorithm}
Let $G(V,E)$ be a graph and let $S(V,W)$ be a biconnected planar spanning subgraph of $G$. There exists an $O(|V|+|E\setminus W|)$-time algorithm that decides whether $\langle G,S \rangle$ admits a 2-bend \cdrawing $\Gamma$ and, in the positive case, computes such a drawing in $O((|V|+|E\setminus W|)^2)$ area.
\end{theorem}

\begin{proof}
First, apply the \textsc{K-Bend-Biconnected-Decision} algorithm to decide whether $\langle G,S \rangle$ satisfies the condition of Theorem~\ref{th:k-bend-biconnected-decision}. If this is the case, let $\cal E$ be the computed embedding of $S$. Then, apply Algorithm~\textsc{2-Bend-Biconnected-Drawing} to construct a grid 2-bend \cdrawing $\Gamma$ of $\langle G,S \rangle$ in which the embedding of $S$ is $\cal E$.

\medskip
\noindent\textsc{Correctness.} By Theorem~\ref{th:k-bend-biconnected-decision}, Algorithm~\textsc{K-Bend-Biconnected-Decision} correctly recognizes instances that do not admit any $k$-bend \cdrawing, and hence any 2-bend \cdrawing of $\langle G,S \rangle$. On the other hand, if Algorithm~\textsc{K-Bend-Biconnected-Decision} gives a positive answer, returning an embedding $\cal E$ with the required properties, Algorithm~\textsc{2-Bend-Biconnected-Drawing} correctly constructs a 2-bend \cdrawing $\Gamma$ of $\langle G,S \rangle$. Indeed, the addition of a cycle $c_f$ and its incident dummy edges inside every face $f$ of the computed embedding of $S$ can always be performed while maintaining the property that the resulting graph is planar and simple. Also, extra dummy edges that make the graph triconnected and still planar can be always added. Denote by $S'$ the resulting triconnected planar graph. Since $S'$ is triconnected, a convex drawing $\Gamma'$ can always be constructed.
Finally, for each edge $e \in E \setminus W$, note that the bends of $e$ are placed on points onto which (dummy) vertices of $S'$ had been placed; hence, there is no overlapping among bend-points and edges. Also, the first and the last straight-line segment of $e$ coincide with edges of $\Gamma'$, while the second straight-line segment of $e$ lies entirely inside the convex face of $\Gamma'$ bounded by the cycle $c_f$ added inside the face $f$ in which $e$ is routed. Hence, crossings in $\Gamma$ only happen inside faces bounded by a cycle $c_f$ for some face $f$. Since such cycles only contain (straight-line segments of) edges of $E \setminus W$, we have that $\Gamma$ is a 2-bend \cdrawing. 

\noindent\textsc{Time and Area Requirements.} By Lemma~\ref{le:k-bend-biconnected-algorithm}, positive or negative instances can be recognized in $O(|V|+|E\setminus W|)$ time. If the instance is positive, graph $S'$ can be constructed and augmented to triconnected in $O(|V|+|E\setminus W|)$ time. Note that, $|S'| = O(|V|+|E\setminus W|)$. Then, a convex drawing $\Gamma'$ of $S'$ can be constructed in $O(|V|+|E\setminus W|)$ time~\cite{br-scdpg-06}. Finally, each edge of $E\setminus W$ can be routed inside a face of $\cal E$ in constant time.
Since the area of $\Gamma$ is the same as the area of $\Gamma'$, and since the area of $\Gamma'$ is $O((|V|+|E\setminus W|)^2)$~\cite{br-scdpg-06}, the area bound follows.
\end{proof}

\section{Conclusions and Open Problems}\label{se:conclusions}

We introduced a new graph drawing problem, i.e., computing a drawing $\Gamma$ of a non-planar graph $G$ such that a desired subgraph $S \subseteq G$ is crossing-free in $\Gamma$. In the setting where edges are straight-line segments, we showed that $\Gamma$ does not always exist even if $S$ is a spanning tree of $G$; also, we provided existential and algorithmic results for meaningful subfamilies of spanning trees and we described a linear-time testing and drawing algorithm when $S$ is a triconnected spanning subgraph. One of the main problems still open in this setting is the following: 

\begin{openproblem}\label{pr:problem-1} Given a graph $G$ and a spanning tree $S$ of $G$, what is the complexity of deciding whether $\langle G,S \rangle$ admits a straight-line \cdrawing?
\end{openproblem}

The problem is still open also when $S$ is a biconnected spanning subgraph (for which one can try to extend the characterization of Theorem~\ref{th:triconnected-decision}) and for subgraphs $S$ that are not necessarily connected and spanning. Hence, the following more general open problem can be stated:

\begin{openproblem}\label{pr:problem-2} 
Given a graph $G$ and a subgraph $S$ of $G$, what is the complexity of deciding whether $\langle G,S \rangle$ admits a straight-line \cdrawing?
\end{openproblem}

Schaefer proposed a variant of Open Problem~\ref{pr:problem-2}, where the edges of $S$ are allowed to cross at most $h$ times~\cite{DBLP:journals/corr/Schaefer13}. In particular, for $h=1$ he proved that the problem turns out to be as hard as the existential theory of the real numbers.

Another intriguing problem, which is more specific, is to extend the results of Lemmas~\ref{le:bad-trees-1} and~\ref{le:bad-trees-2}. Namely:

\begin{openproblem}\label{pr:problem-3} Give a characterization of the pairs $\langle G, S \rangle$ that admit a \cdrawing, when $G$ is a complete graph and $S$ is a spanning tree of~$G$.
\end{openproblem}

In the setting where the edges of $G$ not in $S$ are allowed to bend, we showed that a \cdrawing of $\langle G,S \rangle$ always exists if $S$ is a spanning tree, with different compromises between number of bends per edge and drawing area (see Theorems~\ref{th:trees-1-bend},  \ref{th:trees-3-bend}, and \ref{th:trees-4-bends}). Also, if $S$ is a biconnected spanning subgraph we provided efficient testing and drawing algorithms to compute either 1-bend \cdrawings (with possible exponential area) or 2-bend \cdrawings with polynomial area (see Theorems~\ref{th:1-bend-biconnected-algorithm} and~\ref{th:2-bend-biconnected-algorithm}). However, also in this setting several interesting open problems can be studied; among them we mention the following:

\begin{openproblem}\label{pr:problem-4}
Given a graph $G$, a non-connected subgraph $S$ of $G$, and a positive integer $k$, what is the complexity of deciding whether $\langle G,S \rangle$ admits a $k$-bend \cdrawing? If such a drawing exists, what is the required area?
\end{openproblem}

\begin{openproblem}\label{pr:problem-5}
Given a graph $G$ and a spanning tree $S$ of $G$, does there exist an algorithm that computes a $1$-bend \cdrawing with $o(n^2(n+m))$ area? 
\end{openproblem}

\begin{openproblem}\label{pr:problem-6}
Given a graph $G$ and a biconnected spanning subgraph $S$ of $G$, does there exist an algorithm that computes a $1$-bend \cdrawing of $\langle G,S \rangle$ with polynomial area?
\end{openproblem}

{\small \bibliography{bibliography}}
\bibliographystyle{splncs03}

\end{document}